\documentclass[11pt,draftclsnofoot,twoside,onecolumn,letter]{IEEEtran}

\usepackage{graphicx}
\usepackage{epstopdf}
\DeclareGraphicsExtensions{.eps,.pdf} 
\graphicspath{ {figures/} }
\usepackage{cite}
\usepackage{subfigure}
\usepackage{amssymb,amsmath}
\usepackage{algorithm}
\usepackage{algorithmic}
\usepackage{xcolor}
\usepackage{multirow}
\usepackage{multicol}
\usepackage{cases}
\usepackage{setspace}
\usepackage{amsthm}
\allowdisplaybreaks
\allowdisplaybreaks
\makeatletter
\newcommand\restartchapters{\par
  \setcounter{chapter}{0}%
  \setcounter{section}{0}%
  \gdef\@chapapp{\chaptername}%
  \gdef\thechapter{\@arabic\c@chapter}}
\makeatother

\allowdisplaybreaks

\newtheorem{remark}{{\bf{Remark}}}

\newtheorem{lemma}{\bf {Lemma}}

\newtheorem{proposition}{{\bf Proposition}}

\newcommand{\col}{\mathtt{col}}

\renewcommand{\algorithmicrequire}{\textbf{Input:}}

\newcommand{\ds}{\displaystyle}

\usepackage{capt-of}
\usepackage{dblfloatfix}
\usepackage{varwidth}
\makeatletter
\g@addto@macro\normalsize{%
	\setlength\abovedisplayskip{1pt}
	\setlength\belowdisplayskip{1pt}
	\setlength\abovedisplayshortskip{1.5pt}
	\setlength\belowdisplayshortskip{1.5pt}
}
\makeatother

\newcommand{\subparagraph}{}

\usepackage{titlesec}

\begin{document}
\bstctlcite{IEEEexample:BSTcontrol}

\title{{ Learning-Assisted User Clustering in Cell-Free Massive MIMO-NOMA Networks}}
\author{
	\IEEEauthorblockN{Quang Nhat Le, Van-Dinh Nguyen, Nam-Phong Nguyen, Symeon Chatzinotas, Octavia~A.~Dobre, and Ruiqin Zhao \vspace{-30pt}}
	\\
\thanks{Q. N. Le and O. A. Dobre are with the Dept. of Electrical and Computer Engineering, Memorial University, St. John’s, NL A1B 3X9, Canada (e-mail: \{qnle, odobre\}@mun.ca).}
\thanks{V.-D. Nguyen and S. Chatzinotas are with the Interdisciplinary Centre for Security, Reliability, and Trust (SnT) – University of Luxembourg, L-1855, Luxembourg (e-mail: \{dinh.nguyen, symeon.chatzinotas\}@uni.lu).}
\thanks{N-P. Nguyen is with the School of Electronics and Telecommunications, Hanoi University of Science and Technology, Hanoi, Vietnam (e-mail: phong.nguyennam@hust.edu.vn).}
\thanks{R. Zhao is with the School of Marine Science and Technology, Northwestern Polytechnical University, Xi’an 710072, China (e-mail: rqzhao@nwpu.edu.cn).}
	}

\maketitle
\begin{abstract}
The superior spectral efficiency (SE) and user fairness feature of non-orthogonal multiple access (NOMA) systems are achieved by exploiting user clustering (UC) more efficiently. However, a random UC certainly results in a suboptimal solution while an exhaustive search method comes at the cost of  high complexity, especially for systems of medium-to-large size. To address this problem, we develop two efficient unsupervised machine learning (ML) based UC algorithms, namely k-means++ and improved k-means++, to effectively cluster users into disjoint clusters in cell-free massive multiple-input multiple-output (CFmMIMO) system. Using full-pilot zero-forcing at access points, we derive the sum SE  in closed-form expression taking into account the impact of intra-cluster pilot contamination, inter-cluster interference, and imperfect successive interference cancellation. To comprehensively assess the system performance, we formulate the sum SE optimization problem, and then develop a simple yet efficient iterative algorithm for its solution. In addition, the performance of collocated massive MIMO-NOMA (COmMIMO-NOMA) system is also characterized. Numerical results are provided to show the superior performance of the proposed UC algorithms compared to other baseline schemes. The effectiveness of applying NOMA in CFmMIMO and COmMIMO systems is also validated.  

\end{abstract}

\newpage
\section{Introduction} \label{Introduction} 
The tremendous growth in the number of emerging applications will certainly pose enormous  traffic demands with ultra-high connection density for next-generation wireless networks. It is approximated that more than 19 billion devices are connected to the Internet in 2019, and this number is predicted to exceed 22 billion devices by 2021 \cite{Iotanalytics2016}. The global data traffic of mobile devices is expected to reach 49 exabytes per month by 2021 \cite{Cisco2017_CVNI}, and  will further increase over the next decade. However, traditional orthogonal multiple-access (OMA) techniques seem to reach their fundamental limits in the near future, and therefore are no longer suitable to meet these requirements. Consequently, it calls for innovative techniques that utilize radio resources more efficiently to attain the optimal performance.

Non-orthogonal multiple-access (NOMA) has been envisaged as a key  enabling technology that significantly enhances spectral efficiency (SE) and user fairness of traditional wireless communication systems \cite{Islam:COMSurTutor:2017}. In NOMA, multiple user equipments (UEs) are allowed to simultaneously transmit and receive their signals in the same resources such as time/frequency/code domain by using different signal signatures (i.e., code-domain NOMA) or power levels (i.e., power-domain NOMA) \cite{Islam:COMSurTutor:2017,Dinh:JSAC:Dec2017,HieuTCOM2019}.\footnote{This paper will focus on  power-domain NOMA, which is simply referred to as NOMA for short.} In particular, in a downlink system the key benefit of  NOMA is attributed to the fact that UEs with better channel conditions are able to cancel interference caused by UEs with poorer channel conditions  using successive interference cancellation (SIC) technique. User fairness is then achieved by allocating a large portion of the total power budget to weak UEs, which also guarantees the SIC's feasibility at strong UEs.

Recently, cell-free massive multiple-input multiple-output (CFmMIMO), which is a scalable version of massive MIMO networks, has been introduced to overcome the large propagation losses as well as provide better quality-of-experience services for cell-edge UEs \cite{Ngo:TWC:Mar2017,Nayebi:IEEETWC:Jul2017,Bashar:IEEETWC:Apr2019}. CFmMIMO comprises of a large number of access points (APs) that are spatially distributed over a wide area to coherently  serve multiple UEs in the same time-frequency resources. All APs are coordinated by a central processing unit (CPU) through  fronthaul links. Each AP performs beamforming based on its local channel state information (CSI) only, and this feature thus greatly reduces the complexity in terms of the fronthaul overhead. Since each UE is coherently served by all APs, the effect of cell boundaries can be effectively removed. It was shown in \cite{Ngo:TWC:Mar2017} and \cite{Ngo:IEEETGCN:Mar2018} that  CFmMIMO is superior to small-cell  and  collocated massive MIMO (COmMIMO) in terms of SE and energy efficiency (EE), respectively. However, the key advantages of favorable propagation and  channel hardening properties to multiplex numerous UEs are only achieved in the case of multiple antennas at APs and/or low propagation losses \cite{ChenTWC2018}. From the aforementioned reasons, it is of pivotal interest to study the combination of NOMA and CFmMIMO to reap all their benefits, towards fulfilling the conflicting demands on high SE, massive connectivity with low latency, and high reliability with user fairness of future wireless networks \cite{chen2020massive}.

\subsection{Related Work}
Despite its potential,  there are only a few research works investigating the benefit of NOMA in CFmMIMO systems in the literature. NOMA for the downlink CFmMIMO system was first studied in \cite{LiWCL2018}, where closed-form expression for the achievable sum rate was derived. Numerical results showed the superior performance of NOMA compared to OMA. The authors in \cite{ZhangICC2018} investigated the impact of NOMA  in the uplink CFmMIMO system and derived the closed-form approximation for the sum SE (SSE). Simulation results demonstrated that the CFmMIMO-NOMA system is capable of utilizing the scarce spectrum more efficiently. In \cite{RezaeiTWC2020}, different types of precoding techniques such as maximum ratio transmission (MRT), full-pilot zero-forcing (fpZF), and modified regularized ZF (mRZF) at APs were considered in  downlink CFmMIMO-NOMA systems. It was shown that the downlink CFmMIMO-NOMA system with mRZF and fpZF precoders significantly outperform the OMA with MRT in terms of the achievable sum rate. These  existing works mainly focused on characterizing the performance analysis in CFmMIMO-NOMA systems, but did not show how UEs are paired/grouped.

To be spectrally-efficient, it is crucial to group a sufficiently large number of UEs with distinct channel conditions that performs NOMA jointly \cite{Islam:COMSurTutor:2017,Dinh:JSAC:Dec2017,HieuTCOM2019,Islam:IEEEWirelessComm:Apr2018}. In the context of CFmMIMO-NOMA,  Bashar \textit{et al.} \cite{BasharTCOM2020} proposed three  distance-based  pairing schemes including near pairing, far pairing, and random pairing to group UEs into disjoint clusters. It  is not surprising to see that the close pairing, where two UEs with the smallest distance between them are paired, provides worst performance, which is also aligned with the NOMA principle \cite{Islam:COMSurTutor:2017,Dinh:JSAC:Dec2017}. Another interesting study is to group a large number of UEs into one cluster \cite {RezaeiCL2020}, referred to as user clustering (UC), in which a low complexity suboptimal  method based on the Jaccard distance coefficient was developed to find the most dissimilar UEs in the CFmMIMO-NOMA system. Nevertheless, UC algorithms  in the above-cited works were developed based on  distances among UEs only, while the learning features are missing, resulting in a suboptimal solution.

Recently, unsupervised machine learning (ML) techniques have been considered as an effective means for different optimization targets, which exploit  adaptive learning features. In this regard, the authors in \cite{HeTWC2017} proposed a kernel-power-density based algorithm to cluster  multipath components of MIMO channels into disjoint groups. A novel cluster-based geometrical dynamic stochastic model was proposed in  \cite{XieICC2017},   where scattered nodes were grouped into different clusters according to the density of nodes in MIMO scenarios. In \cite{WangACCESS2019}, a  clustered sparse Bayesian learning algorithm was developed for channel estimation in a hybrid analog-digital massive MIMO system by using the sparsity characteristic of angular domain channel. Notably, the authors in \cite{RozarioICoTCT2018} proposed a novel clustering scheme for machine-to-machine communications in a time-division multiple access-based NOMA system in order to increase the battery lifetime of machines, using the popular k-means algorithm \cite{i14}. This work was extended in \cite{CabreraITNAC2018}  to improve the network sum throughput by considering an enhanced k-means algorithm accompanied NOMA.  Further, the k-means algorithm was used to cluster UEs in  mmwave-NOMA \cite{CuiTWC2018} and CFmMIMO \cite{PalouGC2018}. Although these works demonstrated the effectiveness of applying unsupervised ML in  clustering tasks for various wireless communication systems, its application for UC in CFmMIMO-NOMA has not been previously studied.

\subsection{Motivation and Main Contributions}
In CFmMIMO-NOMA systems, the effects of network interference are increasingly abnormal and acute as the number of APs becomes denser. Most existing works on CFmMIMO-NOMA systems \cite{LiWCL2018,ZhangICC2018,RezaeiTWC2020} focus on the performance analysis while 
they neglect the importance of UC, which has been shown to significantly improve the performance of NOMA-based systems \cite{Islam:COMSurTutor:2017,Dinh:JSAC:Dec2017,CespedesTCOM2020}. A direct application of random UC schemes \cite{Dinh:JSAC:Dec2017,Islam:IEEEWirelessComm:Apr2018} to CFmMIMO-NOMA systems would result in  poor performance, even worse than traditional linear beamforming without NOMA. In addition, a joint UC and beamforming in \cite{HieuTCOM2019}, clustering UEs by means of the tensor model, is not very practical for CFmMIMO-NOMA due to excessively high complexity in terms of computational and signalling overhead.  Although the k-means algorithm has been widely adopted for different clustering tasks \cite{RozarioICoTCT2018,i14,CabreraITNAC2018,CuiTWC2018,PalouGC2018},  its main drawback is sensitivity to the initialization of  centroids. 

%

Taking into account all these issues, in this paper we devise  novel UC algorithms along with an efficient transmission strategy so that the SSE of CFmMIMO-NOMA systems is remarkably enhanced. In particular, our main contributions are summarized as follows:
\begin{itemize}
	\item We propose  two efficient unsupervised ML-based UC algorithms, including
	 k-means++ and improved k-means++, to effectively cluster UEs into disjoint clusters in CFmMIMO-NOMA. The proposed k-means++ algorithms further address  the limitation of k-means due to the randomness of  initial centroids.
	\item Adopting fpZF precoding at APs, we characterize the performance of the proposed CFmMIMO-NOMA system, considering  impacts of intra-cluster pilot contamination, inter-cluster interference, and imperfect SIC. To that end, the closed-form expression of  SSE  is derived. Furthermore, we also present the analytical result for COmMIMO-NOMA, which serves as a benchmark.
	\item To further improve the SSE, we formulate optimization problems for both CFmMIMO-NOMA and COmMIMO-NOMA systems by incorporating power constraints at APs and necessary conditions for implementing SIC at UEs, which belong 	to the difficult class of nonconvex optimization problem. Towards appealing applications, two low-complexity iterative algorithms based on inner approximation (IA) method \cite{Marks:78} are developed for their solutions, which are guaranteed to converge to at least a locally optimal solution.
	\item Extensive numerical results are provided to confirm the effectiveness of the proposed UC algorithms on the SSE performance over the current state-of-the-art approaches (e.g., close-, far- and random-pairing schemes \cite{BasharTCOM2020}, and Jaccard-based UC scheme  \cite{RezaeiCL2020}). They also show the significantly achieved SSE gains of CFmMIMO-NOMA over COmMIMO-NOMA. 
\end{itemize}  

\subsection{Paper Organization and Notations}

The remainder of this paper is organized as follows. Section \ref{sec:sys} describes the system model. In Section \ref{sec:cluster}, two unsupervised ML-based UC algorithms are presented. The performance analysis for CFmMIMO-NOMA is given in Section \ref{sec:PA}. The  proposed iterative algorithms for CFmMIMO-NOMA and COmMIMO-NOMA are provided in Sections  \ref{sec:MPA} and \ref{sec:co}, respectively. Numerical results are given in Section \ref{sec}, while Section \ref{sec:con} concludes the paper.
		
\textit{Notations}: Bold uppercase letters, bold lowercase letters, and lowercase characters stand for matrixes, vectors, and scalars, respectively. $|\cdot|$, $(\cdot)^H$, $(\cdot)^T$, $(\cdot)^*$, and $||\cdot||_2$ correspond to the cardinality, the Hermitian transpose, the transpose, the conjugate, and the $l$2$-$norm operators, respectively. $\mathbb{E}[\cdot]$ represents the expectation operation. $\mathcal{CN}(\mu,\sigma^2)$ stands for circularly symmetric complex Gaussian random variable (RV) with mean $\mu$ and variance $\sigma^2$.

\section{System Model} \label{sec:sys}
\subsection{System Description}

\begin{figure}[!h]
	\centering	\vspace{-10pt}
	\includegraphics[width=0.7\textwidth,trim={0cm 0.0cm 0cm 0.0cm}]{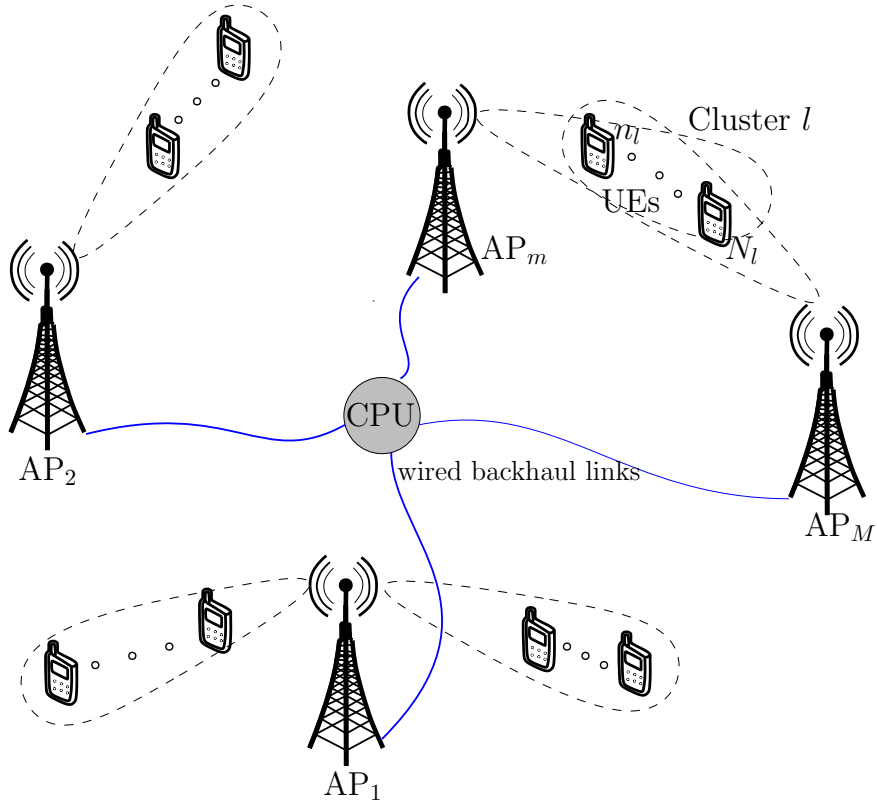}
		\vspace{-5pt}
	\caption{An illustration of the CFmMIMO-NOMA system.}
	\label{fig:sys1}
\end{figure}

We consider an CFmMIMO-NOMA system,  where the set $\mathcal{M}\triangleq\{1,2,\cdots,M\}$ of $M$ APs are connected to the  CPU  through perfect  wired backhaul links to serve the set $\mathcal{N}\triangleq\{1,2,\cdots, N\}$ of $N$ UEs  via a shared wireless medium, as shown in Fig.~\ref{fig:sys1}. Each AP is equipped with $K$  antennas, while each UE has a single antenna.  APs and UEs are assumed to be randomly distributed in a wide coverage area. The communication between APs and UEs follows the time division duplex (TDD) mode. Each coherence interval, denoted by $\tau_c$, includes two phases: uplink training $\tau_p$ ($\tau_p < \tau_c$) and downlink data transmission ($\tau_c - \tau_p$). The total $N$ UEs are grouped into $L$ clusters and each UE belongs to one cluster only. We denote the set of $L$ clusters  by $\mathcal{L}\triangleq\{1,2,\cdots L\}$. The set of UEs in the $l$-th cluster is defined as $\mathcal{N}_l\triangleq\{1_l,\cdots,n_l,\cdots,N_l\}$ with $|\mathcal{N}_l|=N_l$, where $\bigcup_{l\in\mathcal{L}}|\mathcal{N}_l|=N$ and $\mathcal{N}_l\bigcap \mathcal{N}_{l'}=\emptyset$ for $l\neq {l'}$. 

\subsection{Signal Model} 
\subsubsection{Uplink Training}
In the uplink training phase, all  UEs send their training pilots to  APs for channel estimation. Then,  downlink channels are achieved by leveraging the channel reciprocity property of the TDD mode. With the aim of minimizing the channel estimation overhead in CFmMIMO-NOMA,  UEs in the same cluster share the same pilot sequence, and the pilot sequences among different clusters are pairwisely orthogonal \cite{LiWCL2018,BasharTCOM2020} which requires $\tau_p \geq L$. In this paper, we assume that $\tau_p = L$. Let us denote the pilot sequence sent from the UEs in  the $l$-th  cluster by $\boldsymbol{\phi}_{l} \in \mathbb{C}^{\tau_p \times 1}$ with $l \in \{1,2,\ldots,\tau_p\}$, satisfying the orthogonality, i.e., $\|\boldsymbol{\phi}_{l}\|^2_2=\tau_p$ and $\boldsymbol{\phi}^H_l \boldsymbol{\phi}_{l'} = 0$ if $l \neq {l'}$. The channel vector from  UE $n_{l}$ to AP$_m$ is defined as $\textbf{h}_{m,n_{l}} \in \mathbb{C}^{K \times 1}$. In this paper, we focus on slowly time-varying channels, and assume that the channel coefficients are static during the $\tau_c$ interval. The channel $\textbf{h}_{m,n_{l}}$  is generally modeled as follows:
\begin{align}\label{eq:chan}
\textbf{h}_{m,n_{l}} = \sqrt{\beta_{m,n_{l}}} \bar{\textbf{h}}_{m,n_{l}},
\end{align}
where $\beta_{m,n_{l}}$ represents  the large-scale fading coefficient accounting for path loss and shadowing, and $\bar{\textbf{h}}_{m,n_{l}} \in \mathbb{C}^{K \times 1}$ is  small-scale fading vector in which the components are independent and identically distributed (i.i.d.) $\mathcal{CN} (0,1)$ RVs. The training signals received at  AP$_m$ can be written as follows:
\begin{align}\label{eq:Ypm}
\textbf{Y}^p_{m} = \sum\limits_{l\in\mathcal{L}} \sum\limits_{n_{l}\in\mathcal{N}_l} \sqrt{\rho_{n_l}} \textbf{h}_{m,n_{l}} \boldsymbol{\phi}_{l}^H + \textbf{W}^{p}_{m},
\end{align} 
where $\rho_{n_l}$ and $\textbf{W}^{p}_{m} \in \mathbb{C}^{K \times \tau_p}$ are the normalized transmit power of UE ${n_{l}}$ and the additive noise matrix at AP$_m$ whose elements follow $\mathcal{CN} (0,1)$, respectively. 

Given $\textbf{Y}^{p}_{m}$,  AP$_m$ estimates $\textbf{h}_{m,n_{l}}$  using the minimum mean square error (MMSE) criterion. The projection $\hat{\textbf{y}}^{p}_{m} \in \mathbb{C}^{K \times 1}$ of $\textbf{Y}^{p}_{m}$ at  AP$_m$ onto $\boldsymbol{\phi}_{l}$ can be derived as follows:
\begin{align}
\hat{\textbf{y}}^{p}_{m} = \textbf{Y}^{p}_{m} \boldsymbol{\phi}_{l} = \tau_p \sum\limits_{n_{l}\in\mathcal{N}_l} \sqrt{\rho_{n_l}} \textbf{h}_{m,n_{l}} + \textbf{W}^{p}_{m} \boldsymbol{\phi}_{l}.
\end{align}
Hence, the MMSE estimate of $\textbf{h}_{m,n_{l}}$ is given as
\begin{align}\label{eq:hmnl}
\hat{\textbf{h}}_{m,n_{l}} = \frac{\mathbb{E}\{\textbf{h}_{m,n_{l}} (\hat{\textbf{y}}^{p}_{m})^H\}}{\mathbb{E} \{ \hat{\textbf{y}}^{p}_{m} (\hat{\textbf{y}}^{p}_{m})^H \}} \hat{\textbf{y}}^{p}_{m} = \upsilon_{m,n_{l}} \hat{\textbf{y}}^{p}_{m},
\end{align}
where $\ds\upsilon_{m,n_{l}} = \frac{\sqrt{\rho_{n_l}} \beta_{m,n_{l}}}{\tau_p \sum\limits_{{n'}_{l}\in\mathcal{N}_l} \rho_{{n'}_l} \beta_{m,{n'}_{l}} + 1}$.
The estimation error vector of $\textbf{h}_{m,n_{l}}$ is given as
\begin{align}\label{eq:e}
\textbf{e}_{m,n_{l}} = \textbf{h}_{m,n_{l}} - \hat{\textbf{h}}_{m,n_{l}},
\end{align}
where $\textbf{e}_{m,n_{l}}$ and $\hat{\textbf{h}}_{m,n_{l}}$ are i.i.d. RVs distributed as $\mathcal{CN}$ $(\textbf{0},\left(\beta_{m,n_{l}}-\gamma_{m,n_{l}}\right) \textbf{I}_{K})$ and $\mathcal{CN}$ $(\textbf{0},\gamma_{m,n_{l}} \textbf{I}_{K})$, respectively, with $\ds \gamma_{m,n_{l}} = \frac{\tau_p \rho_{n_l} \beta^{2}_{m,n_{l}}}{\tau_p \sum\limits_{{n'}_{l}\in\mathcal{N}_l} \rho_{{n'}_l} \beta_{m,{n'}_{l}} + 1}$. Note that there is no cooperation among APs to exchange the channel estimate information. 

\begin{remark} The so-called pilot contamination exists when  APs estimate the channels of UEs belonging to the same cluster. The relationship of channel estimates of UE ${n_{l}}$ and UE ${{n'}_{l}}$ in the $l$-th cluster with $n_{l} \neq {n'}_l$ and $n_{l}, {n'}_{l} \in \mathcal{N}_l$, at AP$_m$ is expressed as follows:
\begin{align}
\hat{\textbf{h}}_{m,n_{l}} = \frac{\sqrt{\rho_{n_l}} \beta_{m,n_{l}}}{\sqrt{\rho_{{n'}_l}} \beta_{m,{n'}_{l}}} \hat{\textbf{h}}_{m,{n'}_{l}}.
\end{align}
\end{remark}

\subsubsection{Downlink Data Transmission}
Under TDD operation, we consider the channel reciprocity to acquire CSI to precode the transmit signals in the downlink \cite{Ngo:IEEETGCN:Mar2018,Ngo:TWC:Mar2017}. In this paper, we adopt fpZF precoding \cite{InterdonatoGCSIP2018} to cancel inter-cluster interference, but still take into account  intra-cluster interference. Compared with the pure ZF \cite{HieuJSAC2020}, each AP computes fpZF precoding using  its local CSI only, leading to an implementable algorithm. From \eqref{eq:Ypm}, the full-rank matrix $\tilde{\textbf{H}}_m \in \mathbb{C}^{K \times \tau_p}$ of fpZF precoder at AP$_m$ is given by \cite{InterdonatoGCSIP2018}
\begin{align}\label{eq:Hm}
\tilde{\textbf{H}}_m = \textbf{Y}^p_{m} \boldsymbol{\phi},
\end{align}   
where $\boldsymbol{\phi} = [\boldsymbol{\phi}_1,\boldsymbol{\phi}_2,\cdots,\boldsymbol{\phi}_{\tau_p}] \in \mathbb{C}^{\tau_p \times \tau_p}$ denotes the collection of $\tau_p$ orthogonal pilot sequences. Hence, from \eqref{eq:hmnl} and \eqref{eq:Hm}, the channel estimate $\hat{\textbf{h}}_{m,n_{l}}$ is rewritten as
\begin{align}\label{eq:hmnl2}
\hat{\textbf{h}}_{m,n_{l}} = \upsilon_{m,n_{l}} \tilde{\textbf{H}}_m \boldsymbol{\varphi}_l,
\end{align}
where $\boldsymbol{\varphi}_l$ is the $l$-th column of the identity matrix $\textbf{I}_{\tau_p}$. 
From \eqref{eq:Hm} and \eqref{eq:hmnl2}, the beamforming vector $\boldsymbol{w}_{m,l} \in \mathbb{C}^{K \times 1}$ oriented to the $l$-th cluster at  AP$_m$ can be expressed as follows:
\begin{align}\label{eq:omega}
\boldsymbol{w}_{m,l} = \frac{ \tilde{\textbf{H}}_m \bigl( \tilde{\textbf{H}}_m^H \tilde{\textbf{H}}_m \bigr)^{-1} \boldsymbol{\varphi}_l }{\sqrt{\mathbb{E} \left\{ \bigr\|\tilde{\textbf{H}}_m \bigl( \tilde{\textbf{H}}_m^H \tilde{\textbf{H}}_m \bigr)^{-1} \boldsymbol{\varphi}_l \bigl\|^2_2 \right\}}}.
\end{align} 

The transmitted signal $\boldsymbol{x}_m \in \mathbb{C}^{K \times 1}$ from AP$_m$ is given by
\begin{align}\label{eq:xm}
\boldsymbol{x}_m = \sum\limits_{l\in\mathcal{L}} \sum\limits_{n_{l}\in\mathcal{N}_l} \sqrt{\rho^m_{n_{l}}} \boldsymbol{w}_{m,l} x_{n_{l}},
\end{align}
where $x_{n_{l}}$ is the symbol intended for UE $n_l$, and $\rho^m_{n_{l}}$ is the normalized transmit power (normalized by the noise power at AP$_m$) allocated to UE $n_l$ at AP$_m$. Besides, $x_{n_{l}}$ and $x_{{n'}_{l'}}$ for $l, {l'} \in \mathcal{L}$ and $n_l, {n'}_{l'} \in\mathcal{N} $ must satisfy the following condition
\begin{align}
\mathbb{E} \bigl\{ x_{n_{l}} (x_{{n'}_{l'}})^* \bigr\} = \begin{cases}
1, &  \text{if} \ \ l = {l'} \ \ \text{and} \ \ n = {n'},\\
0, & \text{otherwise}.
\end{cases}
\end{align} 
Then, the received signal at  UE ${n_{l}}$ in the $l$-th cluster can be written as
\begin{align}
y_{n_{l}} & = \sum\limits_{m\in\mathcal{M}} \textbf{h}_{m,n_{l}}^H \boldsymbol{x}_m + z_{n_{l}} \nonumber \\
& = \underbrace{\sum\limits_{m\in\mathcal{M}}  \sqrt{\rho^{m}_{n_{l}}}   \textbf{h}_{m,n_{l}}^H \boldsymbol{w}_{m,l} x_{n_{l}}}_{\text{Desired signal}} + \underbrace{\sum\limits_{m\in\mathcal{M}} \sum\limits_{{n'}_{l}\in\mathcal{N}_l \setminus \{n_{l}\}} \sqrt{\rho^{m}_{{n'}_{l}}} \textbf{h}_{m,n_{l}}^H \boldsymbol{w}_{m,l} x_{{n'}_{l}}}_{\text{Intra-cluster interference before SIC}} \nonumber \\
&\quad + \underbrace{\sum\limits_{m\in\mathcal{M}} \sum\limits_{{l'}\in\mathcal{L}\setminus\{l\}} \sum\limits_{n_{l'}\in\mathcal{N}_{l'}} \sqrt{\rho^{m}_{{n}_{l'}}} \textbf{h}_{m,n_{l}}^H \boldsymbol{w}_{m,{l'}} x_{{n}_{l'}}}_{\text{Inter-cluster interference}} + z_{n_{l}},  
\end{align}   
where $z_{n_{l}}\sim\mathcal{CN} (0,1)$ is  the additive white Gaussian noise (AWGN) at UE ${n_{l}}$. 

Without loss of generality, in the $l$-th cluster we consider a descending order of channel gain, i.e., UEs $1_l$ and $N_l$ are the users with strongest and weakest channel gains, respectively. By NOMA principle \cite{Dinh:JSAC:Dec2017,Islam:COMSurTutor:2017}, UE $n_l$ in the $l$-th cluster first decodes the signals of UEs ${n'}_l > n_l$ with poorer channel conditions, and then its own signal is  successively decoded  after removing the interference from those UEs.  Denote by $\text{SINR}^{{n'}_{l}}_{n_{l}}$ and $\text{SINR}^{{n'}_{l}}_{{n'}_{l}}$  the signal-to-interference-plus-noise ratios (SINRs) in decoding the signal of UE ${n'}_{l}$ by UE $n_{l}$ and itself, respectively.
Towards an efficient and implementable SIC, the following necessary condition is considered \cite{BasharTCOM2020}
\begin{align}
&\mathbb{E} \left\{ \text{log}_{2}\bigl(1+\text{SINR}^{{n'}_{l}}_{n_{l}}\bigr) \right\} \geq \mathbb{E} \left\{ \text{log}_{2}\bigr(1+\text{SINR}^{{n'}_{l}}_{{n'}_{l}}\bigl)\right\},\ \forall n_{l} < {n'}_{l}, \forall l\in\mathcal{L}.
\end{align} 

\begin{remark} We note that perfect SIC is practically unattainable owing to the effects of intra-cluster pilot contamination and channel estimation errors. Consequently, the received signal at  UE ${n_{l}}$ in the $l$-th cluster after SIC processing can be written as follows:
\begin{align}\label{eq:ydnl}
\bar{y}_{n_{l}} & = \underbrace{\sum\limits_{m\in\mathcal{M}} \sqrt{\rho^{m}_{n_{l}}} \textbf{h}_{m,n_{l}}^H \boldsymbol{w}_{m,l} x_{n_{l}}}_{\text{Desired signal}} + \underbrace{\sum\limits_{m\in\mathcal{M}} \sum\limits_{{n'}_{l}=1}^{n_{l} - 1} \sqrt{\rho^{m}_{{n'}_{l}}} \textbf{h}_{m,n_{l}}^H \boldsymbol{w}_{m,l} x_{{n'}_{l}}}_{\text{Intra-cluster interference after SIC}} \nonumber \\
&+ \underbrace{\sqrt{\zeta_{{n}_{l}}} \sum\limits_{m \in\mathcal{M}} \sum\limits_{{n''}_{l} = n_{l} + 1}^{N_l} \sqrt{\rho^{m}_{{n''}_{l}}} \textbf{h}_{m,n_{l}}^H \boldsymbol{w}_{m,l} x_{{n''}_{l}} }_{\text{Intra-cluster interference due to imperfect SIC}} + \underbrace{\sum\limits_{m\in\mathcal{M}} \sum\limits_{{l'}\in\mathcal{L}\setminus \{l\}} \sum\limits_{n_{l'}\in\mathcal{N}_{l'}} \sqrt{\rho^{m}_{{n}_{l'}}} \textbf{h}_{m,n_{l}}^H \boldsymbol{w}_{m,{l'}} x_{{n}_{l'}}}_{\text{Inter-cluster interference}} + z_{n_{l}},
\end{align}    
where $\zeta_{{n}_{l}}$ is a general SIC performance coefficient at  UE ${n_{l}}$ in the $l$-th cluster. In particular, $\zeta_{{n}_{l}} = 1$ ($\zeta_{{n}_{l}} = 0$) indicates no SIC (perfect SIC), while $0 < \zeta_{{n}_{l}} < 1$ means imperfect SIC.
\end{remark}

{\color{black}\section{Clustering Cell-Free Massive MIMO-NOMA System}\label{sec:cluster}
In this section, we propose two unsupervised ML-based UC algorithms to effectively divide all UEs into separate clusters, which are done at the CPU by exploiting the large-scale fading coefficients. Similarly to \cite{BasharTCOM2020} and \cite{PalouGC2018},  large-scale fading coefficients of UEs are assumed to be collected and shared with the CPU before performing  UC algorithms. We note that  it is only necessary to estimate the large-scale fading coefficients once every 40 $\tau_c$ intervals \cite{LiWCL2018}, and thus conveying these coefficients via
the backhaul links occurs much less frequently than data transmission. Denote by $\boldsymbol{\beta}_{n} \triangleq [\beta_{1,n},\beta_{2,n},\ldots,\beta_{M,n}]^T \in \mathbb{R}^{M \times 1}$ the set of large-scale fading coefficients from all APs associated to UE $n,\forall n\in\mathcal{N}$. The vector $\boldsymbol{\beta}_{n}$ can be considered as an effective feature-vector denoting the location of UE $n$. 
\subsection{The k-means Algorithm}\label{sec:bac}
The k-means algorithm for UC studied in\cite{CuiTWC2018} and \cite{PalouGC2018} is one of  the simplest unsupervised  ML  algorithms to partition UEs in the coverage area into separate groups.
The key idea is to find a user-specified number of clusters $L$, which are represented by $L$ centroids, one for each cluster. The number of clusters $L$ in the k-means algorithm can be predetermined. The principle of k-means algorithm is given as follows. Firstly, $L$ initial centroids are randomly selected. Secondly, each point is assigned to the nearest centroid, and each mass of points assigned to the same centroid creates a cluster. Then, the centroid of each cluster is updated according to the points associated to the cluster. The assignment and update processes of centroids are repeated until either there is no change in the clusters or  centroids remain similarly.

In the context of CFmMIMO-NOMA, the procedure of k-means can be summarized as follows:
\begin{itemize}
\item Step 1: $L$ initial centroids are randomly selected from $N$ UEs, where $L$ is a predefined number. Let us define the set of $L$ cluster centroids as follows:
\begin{align}
\mathcal{C} = \left\{c_l, l \in \mathcal{L}\right\},
\end{align}
where  $c_l$ represents the centroid of the $l$-th cluster.

\item Step 2: Each UE $n\in\mathcal{N}$ is grouped to the nearest centroid, and hence,  UEs assigned to the same centroid creates a cluster:
\begin{align}\label{eq:Nl}
{l'} = \mathop {\arg \min }\limits_{\forall l\in\mathcal{L}} f_{d}\left(\boldsymbol{\beta}_n,\boldsymbol{\beta}_{c_l}\right),
\end{align}
where $f_{d}\left(\boldsymbol{\beta}_n,\boldsymbol{\beta}_{c_l}\right) =   {\|\boldsymbol{\beta}_n - \boldsymbol{\beta}_{c_l}\|}_2$ represents the Euclidean distance from UE $n$ to centroid $c_l$. As shown in \eqref{eq:Nl}, UE $n$ is grouped to $l'$-th cluster (denoted by centroid $c_{l'}$) since the distance from UE $n$ to centroid $c_{l'}$ is nearest. 

\item Step 3: The centroid of each cluster is recalculated under given UEs assigned to this cluster:
\begin{align}
\boldsymbol{\beta}_{c_l} = \frac{1}{|\mathcal{N}_l|} \sum\limits_{n \in \mathcal{N}_l} \boldsymbol{\beta}_n, \forall l \in \mathcal{L},
\end{align}
where $\boldsymbol{\beta}_{c_l}$ represents the updated centroid for the $l$-th cluster, which can be calculated by the mean of all UEs belonging to the $l$-th cluster.  

\item Step 4: Steps 2-3 are repeated until convergence, i.e.,  there is no change in the clusters or the centroids remain the same.
\end{itemize}

The k-means algorithm for UC in CFmMIMO-NOMA is given in Algorithm \ref{al1}. Note that k-means is a greedy algorithm, which can converge to a local minimum since its performance highly depends on the predefined number of clusters $L$ and the centroid initialization process, i.e., how to select $L$ initial centroids.

\begin{algorithm}[!htbp]
	\begin{algorithmic}[1]
		
		\protect\caption{The k-means Algorithm for UC in CFmMIMO-NOMA.}
		
		\label{al1}

		\STATE \textbf{Input:} $L$ and $\boldsymbol{\beta}_n, \forall n\in\mathcal{N}$.
		
		\STATE //**\textit{Identify $L$ cluster centroids at random $c_l$, $\forall l\in\mathcal{L}$ (Step 1)**}//
		
		\STATE Set $\mathcal{C} = \emptyset$ and $l = 1$, where $\mathcal{C}$ denotes the set of cluster centroids.

		\WHILE {$l \leq L$}
		
		\STATE $c_l = \sf{generateRandom}[1,\textit{N}]$;
		
		\IF {$c_l \not\in \mathcal{C}$}
		
		\STATE $\mathcal{C} \leftarrow c_l$;
		
		\STATE $l = l + 1$;
		
		\ENDIF
		
		\ENDWHILE
		
		\STATE //**\textit{Main process}**//
		
		\WHILE {$\mathcal{C}$ changes}
		
		\STATE //**\textit{Identify $\mathcal{N}_{l'}$, $\forall {l'}\in\mathcal{L}$, containing the subset of UEs that are closer to $c_{l'}$ than $c_l$, with ${l'} \neq l$ (Step 2)}**//
		
		\FOR {$n \in \mathcal{N} \backslash \mathcal{C}$}
		
		\STATE ${l'} = \mathop {\arg \min }\limits_{\forall l\in\mathcal{L}} f_{d}\left(\boldsymbol{\beta}_n,\boldsymbol{\beta}_{c_l}\right)$, where $f_{d}\left(\boldsymbol{\beta}_n,\boldsymbol{\beta}_{c_l}\right)=  {||\boldsymbol{\beta}_n - \boldsymbol{\beta}_{c_l}||}_2$;
		
		\STATE $\mathcal{N}_{l'} \leftarrow n$;
		
		\ENDFOR
		
		\STATE //**\textit{Recalculate $c_l$ of cluster $\mathcal{N}_l$, $\forall l\in\mathcal{L}$ (Step 3)}**//
		
		\FOR {$l = 1:L$}
		
		\STATE $\boldsymbol{\beta}_{c_l} = \frac{1}{|\mathcal{N}_l|} \sum\limits_{n \in \mathcal{N}_l} \boldsymbol{\beta}_n$;
		
		\ENDFOR
		
		\ENDWHILE 
		
		\STATE \textbf{Output:} $\mathcal{N}_l$ and $c_l$, $\forall l\in\mathcal{L}$.
	\end{algorithmic} 
\end{algorithm}

\subsection{Proposed k-means++ Algorithm}\label{sec:ba}
One drawback of the k-means algorithm is that it is sensitive to the initialization of the centroids \cite{ArthurDA2007,FrantiPR2019}. If an initial centroid is a far point, it might not associate with any other points. Equivalently, more than one initial centroids might be created into the same cluster which leads to poor grouping. In this section, the k-means++ algorithm is developed to resolve this issue. It aims at providing a clever initialization of the centroids that improves the quality of the grouping process. Except for the improvement in  the centroid initialization process, the remainder of k-means++ algorithm is the same as in the k-means. In the context of CFmMIMO-NOMA, the procedure of proposed k-means++ can be summarized as follows:
\begin{itemize}
	\item Step 1: The first initial centroid $c_1$ is randomly selected from $N$ UEs.
	
	\item Step 2: For each UE $n$ (with $n \in \mathcal{N} \ \text{and} \ n \not\in \mathcal{C}$), its distance from the nearest centroid is calculated as follows:
	\begin{align}\label{eq:fd}
	f_{d}\left(\boldsymbol{\beta}_n,\boldsymbol{\beta}_{c_t}\right)=  {\|\boldsymbol{\beta}_n - \boldsymbol{\beta}_{c_t}\|}_2,
	\end{align} 
	where ${c_t} = \mathop {\arg \min }\limits_{\forall c_{l}\in\mathcal{C}} f_{d}\left(\boldsymbol{\beta}_n,\boldsymbol{\beta}_{c_l}\right)$. 
	
	\item Step 3: The next centroid is selected from UEs $(\forall n \in \mathcal{N} \backslash \mathcal{C})$ such that the probability of selecting a UE as a centroid is in direct proportion to its distance from the nearest and previously selected centroid, i.e., the UE having the maximum distance from the nearest centroid is virtually to be chosen next as a centroid:
	\begin{align}
	c_l = \mathop {\arg \max }\limits_{\forall n \in \mathcal{N} \backslash \mathcal{C}} f_{d}\left(\boldsymbol{\beta}_n,\boldsymbol{\beta}_{c_t}\right).
	\end{align}
	
	\item Step 4: Steps 2-3 are repeated until $L-1$ centroids are selected.
	
	\item The remaining process follows Steps 2-4 in the k-means algorithm. 
\end{itemize}

The centroid initialization process of the proposed k-means++ ensures that chosen centroids are far away from each other. This increases the opportunity of initially selecting centroids that are located in different clusters. The proposed k-means++ algorithm for UC in CFmMIMO-NOMA is described in Algorithm \ref{al3}.

%
%
%
%
%
%
%
%
%
%
%
%
%
%
%
%
%
%
%
%
%
%
%
%
%
%
%
%
%
%
%

\begin{algorithm}[!htbp]
	\begin{algorithmic}[1]
	\protect\caption{The k-means++ Algorithm for UC in CFmMIMO-NOMA.}
	
	\label{al3}

		\STATE \textbf{Input:} $L$ and $\boldsymbol{\beta}_n$, $\forall n\in\mathcal{N}$.
		
		
		\STATE Set $\mathcal{C} = \emptyset$ and $c_1 = \sf{generateRandom}[1,\textit{N}]$;
		
		\STATE $\mathcal{C} \leftarrow c_1$ and set $f = 0$;
		
		
		\FOR {$l = 2:L$}
		
		\FOR {$n = 1:N$}
		
		\FOR {$t = 1:l-1$}
		
		\IF {$n$ $\neq$ $c_t$}
		
		\STATE $dis\left(1,t\right) = f_{d}\left(\boldsymbol{\beta}_n,\boldsymbol{\beta}_{c_t}\right)$, where $f_{d}\left(\boldsymbol{\beta}_n,\boldsymbol{\beta}_{c_t}\right)=  {\|\boldsymbol{\beta}_n - \boldsymbol{\beta}_{c_t}\|}_2$;
		
		\ELSE
		
		\STATE $dis\left(1,t\right) = \rm{NaN}$;
		
		\STATE $f = f + 1$;
		
		\ENDIF
		
		\ENDFOR
		
		\IF {$f == 0$}
		
		\STATE $dist\left(1,n\right) = \rm{max} \ \textit{dis}$;
		
		\ELSE
		
		\STATE $dist\left(1,n\right) = \rm{NaN}$;
		
		\STATE $f = 0$;
		
		\ENDIF
		
		\ENDFOR
		
		\STATE $c_l = \mathop {\arg \max }\limits_{\forall n \in \mathcal{N} \backslash \mathcal{C}} \textit{dist}$;
		
		\STATE $\mathcal{C} \leftarrow c_l$;
		
		\ENDFOR
		
		
		\WHILE {$\mathcal{C}$ changes}
		
%

		\FOR {$n \in \mathcal{N} \backslash \mathcal{C}$}
		
		\STATE ${l'} = \mathop {\arg \min }\limits_{\forall l\in\mathcal{L}} f_{d}\left(\boldsymbol{\beta}_n,\boldsymbol{\beta}_{c_l}\right)$, where $f_{d}\left(\boldsymbol{\beta}_n,\boldsymbol{\beta}_{c_l}\right)=  {\|\boldsymbol{\beta}_n - \boldsymbol{\beta}_{c_l}\|}_2$;
		
		\STATE $\mathcal{N}_{l'} \leftarrow n$;
		
		\ENDFOR
	
		
		\FOR {$l = 1:L$}
		
		\STATE $\boldsymbol{\beta}_{c_l} = \frac{1}{|\mathcal{N}_l|} \sum\limits_{n \in \mathcal{N}_l} \boldsymbol{\beta}_n$;
		
		\ENDFOR
		
		\ENDWHILE 
		
		\STATE \textbf{Output:} $\mathcal{N}_l$ and $c_l$, $\forall l\in\mathcal{L}$.
		\end{algorithmic} 
\end{algorithm}

\subsection{The Improved k-means++ Algorithm}
As shown in Sections \ref{sec:bac} and \ref{sec:ba}, the performance of the k-means algorithm can be enhanced by  selecting  $L$ initial centroids more effectively. Based on the characteristics of CFmMIMO-NOMA, we propose the improved k-means++ algorithm which includes a new approach to cleverly select  $L$ initial centroids. The procedure of improved k-means++ is summarized as follows:
\begin{itemize}
	\item Step 1: Each AP  identifies an associated UE, denoted by $\Lambda_m$,  which has the best connection, i.e., highest large-scale fading coefficient $\beta_{m,n}$:
	\begin{align}
	\Lambda_m = \mathop {\arg \max }\limits_{\forall n\in\mathcal{N}} \beta_{m,n}, \forall m \in \mathcal{M}.
	\end{align}

	\item Step 2: The CPU then selects a subset of APs, denoted by $\Upsilon_n$, which have best connections to  UE $n$:
	\begin{align}
	\Upsilon_n = \left\{ \text{AP}_{m}: \text{UE} \ {n} == \Lambda_m \right\}, \forall n \in \mathcal{N}.
	\end{align}

	\item Step 3: The CPU selects a UE having the highest number of serving APs as a centroid:
	\begin{align}
	c_l = \mathop {\arg \max }\limits_{\forall n \in \mathcal{N} \backslash \mathcal{C}} |\Upsilon_n|,
	\end{align} 
	where $|\Upsilon_n|$ denotes the cardinality of $\Upsilon_n$.
	
	\item Step 4: Step 3 is repeated until $L$ centroids are chosen.
	
	\item The remaining process follows Steps 2-4 in the k-means algorithm. 

\end{itemize}

The centroid initialization process of the improved k-means++ for UC in CFmMIMO-NOMA is described in Algorithm \ref{al4}.

\begin{algorithm}[!htbp]
	\begin{algorithmic}[1]
	
	\protect\caption{Centroid Initialization Process of the Improved k-means++ Algorithm for UC in CFmMIMO-NOMA.}
	
	\label{al4}
		
		\STATE \textbf{Input:} $L$ and $\boldsymbol{\beta}_n, \forall n\in\mathcal{N}$.
		
		
		\FOR {$m = 1:M$}
		
		\STATE $\Lambda_m = \mathop {\arg \max }\limits_{\forall n\in\mathcal{N}} \beta_{m,n}$;
		
		\ENDFOR	   
		
		
		\FOR {$n = 1:N$}
		
		\FOR {$m = 1:M$}
		
		\IF {$n == \Lambda_m$}
		
		\STATE $\Upsilon_n \leftarrow m$;
		
		\ENDIF
		
		\ENDFOR
		
		\ENDFOR
		
		
		\STATE $\mathcal{C} = \emptyset$, where $\mathcal{C}$ denotes the set of cluster centroids.
		
		\FOR {$l = 1:L$}
		
		\STATE $c_l = \mathop {\arg \max }\limits_{\forall n \in \mathcal{N} \backslash \mathcal{C}} |\Upsilon_n|$;
		
		\STATE $\mathcal{C} \leftarrow c_l$;
		
		\ENDFOR
		
		\STATE \textbf{Output:} $\mathcal{C}$.
		
	\end{algorithmic}
	
\end{algorithm}

\subsection{Complexity Analysis}
As shown in \cite{CuiTWC2018}, the complexity of the k-means algorithm is $\mathcal{O}\left(N L I M\right)$, where $I$ denotes the total number of iterations until convergence.
We recall that compared to the k-means, the k-means++ and  improved k-means++ algorithms require the modification of  centroid initialization process.  
All centroids in the k-means algorithm are randomly chosen, which leads to the computational complexity of $\mathcal{O}\left(N\right)$. The proposed k-means++ algorithm has to make a full search through all UEs for every centroid sampled, resulting  to the complexity of $\mathcal{O}\left(N L M\right)$ \cite{BachemAI2016}. Similarly, the complexity of the improved k-means++ algorithm is $\mathcal{O}\left(M N + N M + L N\right)$, which is lower than that of the k-means++. In addition, although the centroid initialization process in the proposed k-means++ algorithms is computationally more expensive than the original k-means, the performance of the former is much better than the latter. This will be elaborated in Section \ref{sec}.}

\section{Performance Analysis} \label{sec:PA}
Given the UC algorithms in Section \ref{sec:cluster}, we now derive the SSE of CFmMIMO-NOMA. From \eqref{eq:ydnl}, the $\text{SINR}$ of UE $n_{l}$ in the $l$-{th} cluster is given as
\begin{equation} \label{eq:SINR}
\text{SINR}_{n_{l}} = \frac{|\text{DS}_{n_{l}}|^2}{\mathbb{E} \left\{|\text{BU}_{n_{l}}|^2\right\} + \sum\limits_{{n'}_{l}=1}^{n_{l} - 1} \mathbb{E} \left\{|\text{ICI}_{n_{l}}|^2\right\} + \sum\limits_{{n''}_{l} = n_{l} + 1}^{N_l} \mathbb{E} \left\{|\text{RICI}_{n_{l}}|^2\right\} + \sum\limits_{{l'}\in\mathcal{L}\setminus \{l\}} \sum\limits_{n_{l'}\in\mathcal{N}_{l'}} \mathbb{E} \left\{|\text{UI}_{n_{l}}|^2\right\} + 1},
\end{equation}
where $\text{DS}_{n_{l}} = \mathbb{E} \Bigl\{\sum\limits_{m\in\mathcal{M}} \sqrt{\rho^{m}_{n_{l}}} \textbf{h}_{m,n_{l}}^H \boldsymbol{w}_{m,l}\Bigr\}$, $\text{BU}_{n_{l}} = \Bigl(\sum\limits_{m\in\mathcal{M}} \sqrt{\rho^{m}_{n_{l}}} \textbf{h}_{m,n_{l}}^H \boldsymbol{w}_{m,l} - \mathbb{E} \Bigl\{\sum\limits_{m\in\mathcal{M}} \sqrt{\rho^{m}_{n_{l}}} \textbf{h}_{m,n_{l}}^H \boldsymbol{w}_{m,l}\Bigl\}\Bigl)$, $\text{ICI}_{n_{l}} = \sum\limits_{m\in\mathcal{M}} \sqrt{\rho^{m}_{{n'}_{l}}} \textbf{h}_{m,n_{l}}^H \boldsymbol{w}_{m,l}$, $\text{RICI}_{n_{l}} = \sqrt{\zeta_{{n}_{l}}} \sum\limits_{m\in\mathcal{M}} \sqrt{\rho^{m}_{{n''}_{l}}} \textbf{h}_{m,n_{l}}^H \boldsymbol{w}_{m,l}$, and $\text{UI}_{n_{l}} = \sum\limits_{m\in\mathcal{M}} \sqrt{\rho^{m}_{{n}_{l'}}} \textbf{h}_{m,n_{l}}^H \boldsymbol{w}_{m,{l'}}$ are coherent beamforming gain (desired signal), beamforming gain uncertainty, intra-cluster interference after SIC, residual interference due to imperfect SIC, and inter-cluster interference, respectively.

To simplify \eqref{eq:SINR}, we first compute the expectation term in the denominator of \eqref{eq:omega} \cite{Tulino2004}:
\begin{align}\label{eq:E}
\mathbb{E} \left\{ \bigr\|\tilde{\textbf{H}}_m \bigl( \tilde{\textbf{H}}_m^H \tilde{\textbf{H}}_m \bigr)^{-1} \boldsymbol{\varphi}_l \bigl\|^2_2 \right\} = \frac{\upsilon^2_{m,n_{l}}}{\gamma_{m,n_{l}} (K-\tau_p)},\ \forall n_{l}\in\mathcal{N}_{l}.
\end{align}

From \eqref{eq:hmnl2}, \eqref{eq:omega}, and \eqref{eq:E}, we have
\begin{align}\label{eq:hmnl3}
{\hat{\textbf{h}}}^H_{m,n_{i}} \boldsymbol{w}_{m,l} &= \frac{\upsilon_{m,n_{i}}}{\upsilon_{m,n_{l}}} \boldsymbol{\varphi}^H_i \boldsymbol{\varphi}_l \sqrt{\gamma_{m,n_{l}} (K-\tau_p)} \nonumber \\
& = \begin{cases}
\sqrt{\gamma_{m,n_{l}} (K-\tau_p)}, &  \text{if} \ \ {i = l},\\
0, & \text{if} \ \ i \neq l.
\end{cases}
\end{align}

\begin{lemma}\label{lemma1} The closed-form expression for the SE of UE ${n_{l}}$ in the $l$-th cluster is given by
\begin{align}\label{eq:Rate}
R_{n_{l}} &= \Bigl(1-\frac{\tau_p}{\tau_c}\Bigr) \log_{2}\Bigl(1+\mathrm{SINR}_{n_{l}}\Bigl) \nonumber \\
& = \Bigl(1-\frac{\tau_p}{\tau_c}\Bigl) \log_{2}\Bigl(1+\min_{{n'}_{l} = 1,\ldots,n_{l}}\mathrm{SINR}^{n_{l}}_{{n'}_{l}}\Bigl),\ \forall n_l.
\end{align}
 By $\boldsymbol{\rho}\triangleq\{\rho^m_{n_l}\}_{m\in\mathcal{M}, n_l\in\mathcal{N}_l, l\in\mathcal{L}}$, $\mathrm{SINR}^{n_{l}}_{n_{l}}$ and $\mathrm{SINR}^{n_{l}}_{{n'}_{l}}$, $\forall {n'}_{l} < n_{l}$,  are derived as follows:
\begin{align}\label{eq:SINR1}
\mathrm{SINR}^{n_{l}}_{n_{l}} = \frac{(K-\tau_p) \Bigl(\sum\limits_{m\in\mathcal{M}} \sqrt{\rho^{m}_{n_{l}} \gamma_{m,n_{l}}}\Bigl)^2}{\mathcal{I}^{n_{l}}_{n_{l}}(\boldsymbol{\rho}) + 1},
\end{align}
\begin{align}\label{eq:SINR2}
\mathrm{SINR}^{n_{l}}_{{n'}_{l}} = \frac{(K-\tau_p) \Bigl(\sum\limits_{m\in\mathcal{M}} \sqrt{\rho^{m}_{n_{l}} \gamma_{m,{n'}_{l}}}\Bigr)^2}{\mathcal{I}^{n_{l}}_{{n'}_{l}}(\boldsymbol{\rho}) + 1},
\end{align}
where $\mathcal{I}^{n_{l}}_{n_{l}}(\boldsymbol{\rho})$ and $\mathcal{I}^{n_{l}}_{{n'}_{l}}(\boldsymbol{\rho})$ are defined as
\begin{align}\label{eq:SINR2aa}
\mathcal{I}^{n_{l}}_{n_{l}}(\boldsymbol{\rho})&\triangleq\sum\limits_{{n''}_{l}\in\mathcal{N}_{l}\setminus \{n_{l}\}} \eta_{{n''}_{l'}} (K-\tau_p) \Bigl(\sum\limits_{m\in\mathcal{M}} \sqrt{\rho^{m}_{{n''}_{l}} \gamma_{m,n_{l}}}\Bigr)^2 \nonumber\\
&+ \sum\limits_{{l'}\in\mathcal{L}} \sum\limits_{{n''}_{l'}\in\mathcal{N}_{l'}} \sum\limits_{m\in\mathcal{M}} \eta_{{n''}_{l'}} \rho^{m}_{{n''}_{l'}} \left(\beta_{m,n_{l}}-\gamma_{m,n_{l}}\right),\nonumber\\
\mathcal{I}^{n_{l}}_{{n'}_{l}}(\boldsymbol{\rho})&\triangleq\sum\limits_{{n''}_{l}\in\mathcal{N}_{l}\setminus \{n_{l}\}} \eta_{{n''}_{l'}} (K-\tau_p) \Bigr(\sum\limits_{m\in\mathcal{M}} \sqrt{\rho^{m}_{{n''}_{l}} \gamma_{m,{n'}_{l}}}\Bigl)^2 \nonumber\\
&+ \sum\limits_{{l'}\in\mathcal{L}} \sum\limits_{{n''}_{l'}\in\mathcal{N}_{l'}} \sum\limits_{m\in\mathcal{M}} \eta_{{n''}_{l'}} \rho^{m}_{{n''}_{l'}} \left(\beta_{m,{n'}_{l}}-\gamma_{m,{n'}_{l}}\right),\nonumber
\end{align}
with 
\begin{equation}
\eta_{{n''}_{l'}} =
\begin{cases}
1, & \text{if} \ {l'} \ \neq \ l \ \text{or} \ {l'} \ = \ l \ \text{and} \ {n''}_{l} \ \leq \ {n}_{l},\nonumber\\
\zeta_{{n}_{l}}, & \text{otherwise}.\nonumber\\
\end{cases}       
\end{equation}

\end{lemma}
\begin{proof}
The proof is given in Appendix \ref{ap1}.
\end{proof}

 We define the virtual channel of UE ${n_l}$ in the $l$-th cluster as $\textbf{h}_{n_{l}} = \left[\gamma_{1,n_{l}},\ldots,\gamma_{M,n_{l}}\right]^T$, $\forall n_l \in \mathcal{N}_l$. We assume that UEs in the $l$-th cluster are sorted based on their virtual channels, such as $\|\textbf{h}_{1_l}\|_2 \geq \|\textbf{h}_{2_l}\|_2 \geq \ldots \geq \|\textbf{h}_{N_l}\|_2$, $\forall l \in \mathcal{L}$.  From  \eqref{eq:Rate}, the SSE of CFmMIMO-NOMA is expressed as
\begin{align}
R_{\Sigma} = \sum\limits_{l \in \mathcal{L}} \sum\limits_{n_{l} \in \mathcal{N}_{l}} R_{n_{l}} = \Bigl(1-\frac{\tau_p}{\tau_c}\Bigr) \sum\limits_{l \in \mathcal{L}} \sum\limits_{n_{l} \in \mathcal{N}_{l}}  \log_{2}\Bigl(1+\mathrm{SINR}_{n_{l}}\Bigl).
\end{align}

 From \eqref{eq:SINR1} and \eqref{eq:SINR2}, it is clear that the SSE of CFmMIMO-NOMA  highly depends on the power allocation (PA) at all APs. Thus, it is necessary to optimize the transmit power at APs so that the SSE of CFmMIMO-NOMA can be enhanced, which will be detailed next.

\section{The Sum Spectral Efficiency Maximization}\label{sec:MPA}

We aim at optimizing the normalized transmit power  $\boldsymbol{\rho}\triangleq\{\rho^m_{n_l}\}_{m, n_l, l}$ to maximize the SSE under the constraints of the transmit power budget at the APs and SIC conditions. The optimization problem can be mathematically expressed as
\begin{subequations}\label{eq:opp}
	\begin{align}
	 \underset{\boldsymbol{\rho}}{\max} \label{eq:op1a}
	 & \quad\Bigl(1-\frac{\tau_p}{\tau_c}\Bigr) \sum\limits_{l \in \mathcal{L}} \sum\limits_{n_{l} \in \mathcal{N}_{l}}  \log_{2}\bigl(1+\mathrm{SINR}_{n_{l}}\bigl) \\ \label{eq:op1b}
	 \mathsf{s.t.}&\quad    \sum\limits_{l \in \mathcal{L}} \sum\limits_{n_{l} \in \mathcal{N}_l} \rho^m_{n_{l}} \leq P^{m}_{\text{max}}, \forall m \in \mathcal{M},\\ \label{eq:op1c}
	&\quad   \rho^m_{n_l} \leq \rho^m_{n_l+1}, n_l \in \left[1,N_l-1\right], \forall m \in \mathcal{M},  l \in \mathcal{L}.
	\end{align}
\end{subequations}

Herein, constraint \eqref{eq:op1b} indicates that the total transmit power at AP$_{m}$ is limited by the normalized maximum power $P^{m}_{\text{max}}$, while constraint \eqref{eq:op1c} is the necessary condition to implement SIC in the $l$-th cluster, $\forall l \in \mathcal{L}$.
We note that $\mathrm{SINR}_{n_{l}}$ in \eqref{eq:op1a} is a nonconvex and nonsmooth function with respect to $\boldsymbol{\rho}$, making problem \eqref{eq:opp} intractable.
Therefore, it may not be possible to solve the problem directly. In addition, the globally optimal solution (e.g., exhaustive search) comes at the cost of high computational complexity, and may not be suitable for practical implementation. In what follows, we develop newly approximated functions using the IA framework \cite{Marks:78,Beck:JGO:10}, and then propose a fast converging and low-complexity algorithm.

\textit{Equivalent Optimization Problem:} 
To apply the IA method, several transformations are necessary to make \eqref{eq:opp} tractable. To do so, we introduce the auxiliary variables $\mathbf{r}\triangleq\big\{r_{n_{l}}\big\}_{\forall n_{l}}$ and $\boldsymbol{\varphi}\triangleq\big\{\varphi_{n_{l}}\big\}_{\forall n_{l}}$ to  rewrite \eqref{eq:opp} equivalently as
\begin{subequations}\label{eq:op2} 
	\begin{align}
	\underset{\boldsymbol{\rho}, \mathbf{r}, \boldsymbol{\varphi}}{\max} \label{eq:op2a}
	&\quad \Bigl(1-\frac{\tau_p}{\tau_c}\Bigr) \sum\limits_{l \in \mathcal{L}} \sum\limits_{n_{l} \in \mathcal{N}_l} r_{n_{l}} \\ 
	\mathsf{s.t.} &\quad  \ln\left(1+\varphi_{n_{l}}\right) \geq r_{n_{l}}\ln2,\ \forall n_{l} \in \mathcal{N}_l,   \label{eq:op2b}\\ 
	&\quad \text{SINR}^{n_{l}}_{{n'}_{l}} \geq  \varphi_{n_{l}},\ \forall {n'}_{l} < n_{l},\ \forall n_{l} \in \mathcal{N}_l, \label{eq:op2c}\\
	&\quad \text{SINR}^{n_{l}}_{n_{l}} \geq  \varphi_{n_{l}},\ \forall n_{l} \in \mathcal{N}_l, \label{eq:op2d}\\
	&\quad  \eqref{eq:op1b}, \eqref{eq:op1c}\label{eq:op2e}.
	\end{align}
\end{subequations} 

It is clear that the objective function becomes linear. The equivalence between \eqref{eq:opp} and \eqref{eq:op2} is verified by the following lemma.

\begin{lemma}\label{lemma2}
Problems \eqref{eq:opp} and \eqref{eq:op2} share the same optimal solution set and the same optimal objective value. In particular, let $(\boldsymbol{\rho}^\star, \mathbf{r}^\star, \boldsymbol{\varphi}^\star)$ be the optimal solution to problem \eqref{eq:op2},  then $\boldsymbol{\rho}^\star$ is also the optimal solution to problem  \eqref{eq:opp} and vice versa.
\end{lemma}
\begin{proof}
The proof is done by showing the fact that constraints \eqref{eq:op2b}-\eqref{eq:op2d} will hold with equality at the optimum. We prove this statement by contradiction. Suppose that constraints \eqref{eq:op2c} and \eqref{eq:op2d} are inactive at the optimum for some users, i.e., there exists $\varphi'_{n_{l}} > 0$ such as $\min\bigl(\text{SINR}^{n_{l}}_{{n'}_{l}},  \text{SINR}^{n_{l}}_{n_{l}}\bigr) = \varphi'_{n_{l}} > \varphi^{\star}_{n_{l}}$. It is clear that  $\varphi'_{n_{l}}$ is also a feasible point to \eqref{eq:op2}, and $r'_{n_{l}} = \ln\left(1+\varphi'_{n_{l}}\right)   > \ln\left(1+\varphi^{\star}_{n_{l}}\right) = r^{\star}_{n_{l}}$. As a consequence, this results in a strictly larger objective value, i.e., $\bigl(1-\frac{\tau_p}{\tau_c}\bigr) \sum\limits_{l \in \mathcal{L}} \sum\limits_{n_{l} \in \mathcal{N}_l} r'_{n_{l}} > \bigl(1-\frac{\tau_p}{\tau_c}\bigr) \sum\limits_{l \in \mathcal{L}} \sum\limits_{n_{l} \in \mathcal{N}_l} r^{\star}_{n_{l}}$, which contradicts the assumption that
$(\boldsymbol{\rho}^\star, \mathbf{r}^\star, \boldsymbol{\varphi}^\star)$ represent the optimal solution to problem \eqref{eq:op2}.
\end{proof}

\textit{Inner Approximation (IA) for Problem \eqref{eq:op2}:} The nonconvex parts include \eqref{eq:op2c} and \eqref{eq:op2d}. The direct application of IA method is still not possible due to the complication of $\text{SINR}^{n_{l}}_{{n'}_{l}}$ and $\text{SINR}^{n_{l}}_{n_{l}}$. In the following, we make the change of variable as $\rho^m_{n_{l}}=(\hat{\rho}^m_{n_{l}})^2, \forall n_l \in \mathcal{N}_l$. Let us handle \eqref{eq:op2c} first by rewriting $\text{SINR}^{n_{l}}_{{n'}_{l}}$ as
\begin{align}\label{eq:SINR2:change}
\mathrm{SINR}^{n_{l}}_{{n'}_{l}} = \frac{(K-\tau_p) \bigl(\sum\limits_{m\in\mathcal{M}} \hat{\rho}^m_{n_{l}} \sqrt{ \gamma_{m,{n'}_{l}}}\bigr)^2}{\mathcal{I}^{n_{l}}_{{n'}_{l}}(\hat{\boldsymbol{\rho}}) + 1},
\end{align}
where $\hat{\boldsymbol{\rho}}\triangleq\{\hat{\rho}^m_{n_{l}}\}_{\forall n_l}$ and  $\mathcal{I}^{n_{l}}_{{n'}_{l}}(\hat{\boldsymbol{\rho}})\triangleq\sum\limits_{{n''}_{l}\in\mathcal{N}_{l}\setminus \{n_{l}\}} \eta_{{n''}_{l'}} (K-\tau_p) \bigr(\sum\limits_{m\in\mathcal{M}} \hat{\rho}^{m}_{{n''}_{l}}\sqrt{ \gamma_{m,{n'}_{l}}}\bigl)^2 
+ \sum\limits_{{l'}\in\mathcal{L}} \sum\limits_{{n''}_{l'}\in\mathcal{N}_{l'}} \sum\limits_{m\in\mathcal{M}}$ $ \eta_{{n''}_{l'}} \left(\hat{\rho}^{m}_{{n''}_{l'}}\right)^2 \bigl(\beta_{m,{n'}_{l}}$ - $ \gamma_{m,{n'}_{l}}\bigr).$
%
By introducing the slack variables $\boldsymbol{\varpi}\triangleq \{\varpi^{n_{l}}_{{n}_{l}}\}_{\forall n_{l}}, \boldsymbol{\tau}\triangleq\{\tau^{{n}_{l}}_{{n}_{l}}\}_{\forall n_{l}}$,	and 
$\boldsymbol{\theta}\triangleq\{\theta^{n_{l}}_{{n}_{l}}\}_{\forall n_{l}}$, constraint \eqref{eq:op2c}
can be equivalently rewritten as
\begin{subnumcases}{\label{eq:op2c1a}
	\eqref{eq:op2c} \Leftrightarrow} 
\sum\limits_{m \in \mathcal{M}} \hat{\rho}^{m}_{n_{l}}\sqrt{ \gamma_{m,{n'}_{l}}} \geq \varpi^{n_{l}}_{{n'}_{l}},\ \forall {n'}_{l} < n_{l}, \forall n_{l} \in \mathcal{N}_l, \IEEEyessubnumber\label{eq:op2c1b}\\
\sum\limits_{m \in \mathcal{M}} \hat{\rho}^{m}_{{n''}_{l}}\sqrt{ \gamma_{m,{n'}_{l}}} \leq \tau^{{n'}_{l}}_{{n''}_{l}},\ \forall {n'}_{l} < n_{l},\ \forall n_{l} \in \mathcal{N}_l,\IEEEyessubnumber\label{eq:op2c1c}\\
\mathcal{I}^{n_{l}}_{{n'}_{l}}(\hat{\boldsymbol{\rho}},\boldsymbol{\tau}) \leq \theta^{n_{l}}_{{n'}_{l}},\ \forall {n'}_{l} < n_{l},\ \forall n_{l} \in \mathcal{N}_l,\qquad\label{eq:op2c1d}\\
(K-\tau_p)\frac{ \bigr(\varpi^{n_{l}}_{{n'}_{l}}\bigl)^2}{\theta^{n_{l}}_{{n'}_{l}}+1} \geq \varphi_{n_{l}},\ \forall {n'}_{l} < n_{l},\ \forall n_{l} \in \mathcal{N}_l,\label{eq:op2c1e}
\end{subnumcases}
where $\mathcal{I}^{n_{l}}_{{n'}_{l}}(\hat{\boldsymbol{\rho}},\boldsymbol{\tau})\triangleq\sum\limits_{{n''}_{l}\in\mathcal{N}_{l}\setminus \{n_{l}\}} \eta_{{n''}_{l'}} (K-\tau_p) \bigr(\tau^{{n'}_{l}}_{{n''}_{l}}\bigl)^2 
+ \sum\limits_{{l'}\in\mathcal{L}} \sum\limits_{{n''}_{l'}\in\mathcal{N}_{l'}} \sum\limits_{m\in\mathcal{M}}$ $ \eta_{{n''}_{l'}} \bigr(\hat{\rho}^{m}_{{n''}_{l'}}\bigl)^2 \bigl(\beta_{m,{n'}_{l}}$ - $ \gamma_{m,{n'}_{l}}\bigr)$ is a quadratic function. Here, constraint \eqref{eq:op2c1e} remains nonconvex. We note that $ (\varpi^{n_{l}}_{{n'}_{l}})^2/(\theta^{n_{l}}_{{n'}_{l}}+1)$ is the quadratic-over-linear function, which is convex with respect to $(\varpi^{n_{l}}_{{n'}_{l}}, \theta^{n_{l}}_{{n'}_{l}})$. Let $(\varpi^{n_{l},(\kappa)}_{{n'}_{l}}, \theta^{n_{l},(\kappa)}_{{n'}_{l}})$ be a feasible point of $(\varpi^{n_{l}}_{{n'}_{l}}, \theta^{n_{l}}_{{n'}_{l}})$ at the $\kappa$-th iteration of an iterative algorithm and by the IA principle, constraint \eqref{eq:op2c1e}  can be convexified as
\begin{align}\label{eq:op2c1eConvex}
(K-\tau_p)\Bigl(\frac{2{\varpi^{n_{l},(\kappa)}_{{n'}_{l}}} }{{\theta^{n_{l},(\kappa)}_{{n'}_{l}}}+1}\varpi^{n_{l}}_{{n'}_{l}}-\frac{\bigl({\varpi^{n_{l},(\kappa)}_{{n'}_{l}}}\bigr)^2}{\bigl({\theta^{n_{l},(\kappa)}_{{n'}_{l}}}+1\bigr)^2}(\theta^{n_{l}}_{{n'}_{l}}+1)\Bigr) \geq \varphi_{n_{l}},\ \forall {n'}_{l} < n_{l},\ \forall n_{l} \in \mathcal{N}_l.
\end{align}

 Similarly, constraint \eqref{eq:op2d} can be iteratively approximated as
\begin{IEEEeqnarray}{rCl} 
&&\sum\limits_{m \in \mathcal{M}} \hat{\rho}^{m}_{n_{l}}\sqrt{ \gamma_{m,n_{l}}}  \geq \varpi^{n_{l}}_{n_{l}},\ \forall n_{l} \in \mathcal{N}_l, \IEEEyessubnumber\label{eq:op2d1b}\\
&&\sum\limits_{m \in \mathcal{M}} \hat{\rho}^{m}_{{n''}_{l}}\sqrt{ \gamma_{m,n_{l}}} \leq \tau^{n_{l}}_{{n''}_{l}},\ \forall n_{l} \in \mathcal{N}_l,\IEEEyessubnumber\label{eq:op2d1c}\\
&&\mathcal{I}^{n_{l}}_{{n}_{l}}(\hat{\boldsymbol{\rho}},\boldsymbol{\tau})
\leq \theta^{n_{l}}_{n_{l}},\ \forall n_{l} \in \mathcal{N}_l,\qquad\IEEEyessubnumber\label{eq:op2d1d}\\
&&(K-\tau_p) \Bigl(\frac{2{\varpi^{n_{l},(\kappa)}_{n_{l}}} }{{\theta^{n_{l},(\kappa)}_{n_{l}}} +1}\varpi^{n_{l}}_{n_{l}}-\frac{\bigl({\varpi^{n_{l},(\kappa)}_{n_{l}}}\bigr)^2}{\bigl({\theta^{n_{l},(\kappa)}_{n_{l}}} + 1\bigr)^2}(\theta^{n_{l}}_{n_{l}}+1)\Bigr) \geq \varphi_{n_{l}},\ \forall n_{l} \in \mathcal{N}_l,\IEEEyessubnumber\label{eq:op2d1e}
\end{IEEEeqnarray}
where $\mathcal{I}^{n_{l}}_{{n}_{l}}(\hat{\boldsymbol{\rho}},\boldsymbol{\tau})
\triangleq\sum\limits_{{n''}_{l}\in\mathcal{N}_{l}\setminus \{n_{l}\}} \eta_{{n''}_{l'}} (K-\tau_p) \left(\tau^{n_{l}}_{{n''}_{l}}\right)^2 + \sum\limits_{{l'} \in \mathcal{L}} \sum\limits_{{n''}_{l'} \in \mathcal{N}_{l'}} \sum\limits_{m \in \mathcal{M}} \eta_{{n''}_{l'}} \left(\hat{\rho}^{m}_{{n''}_{l'}}\right)^2 \left(\beta_{m,n_{l}}-\gamma_{m,n_{l}}\right)$.

In summary, the convex approximate program of \eqref{eq:op2} solved at iteration $\kappa+1$ is given as
\begin{subequations}\label{eq:convexprogramlog}
	\begin{align}
	& \underset{\hat{\boldsymbol{\rho}}, \mathbf{r}, \boldsymbol{\varphi},\boldsymbol{\varpi},\boldsymbol{\tau},\boldsymbol{\theta} }{\max} 
	\quad \Bigl(1-\frac{\tau_p}{\tau_c}\Bigr) \sum\limits_{l \in \mathcal{L}} \sum\limits_{n_{l} \in \mathcal{N}_{l}} r_{n_{l}} \label{eq:eq:convexprogramlog:a}\\ 
	&\quad\ \textsf{s.t.} \qquad  \eqref{eq:op2b}, \eqref{eq:op2c1b}{-}\eqref{eq:op2c1d},
	  \eqref{eq:op2c1eConvex}, 
	   \eqref{eq:op2d1b}{-}\eqref{eq:op2d1e},\label{eq:eq:convexprogramlog:b}\\
	&\qquad\qquad\    \sum\limits_{l \in \mathcal{L}} \sum\limits_{n_{l} \in \mathcal{N}_l} (\hat{\rho}^m_{n_{l}})^2 \leq P^{m}_{\text{max}}, \forall m \in \mathcal{M},\label{eq:eq:convexprogramlog:c}\\ 
	&\qquad\qquad\   \hat{\rho}^m_{n_l} \leq \hat{\rho}^m_{n_l+1}, n_l \in \left[1,N_l-1\right], \forall m \in \mathcal{M},  l \in \mathcal{L}. \label{eq:eq:convexprogramlog:d}  
	\end{align}
\end{subequations}	

\textit{Conic Quadratic Program}: Problem \eqref{eq:convexprogramlog} is a mix of exponential and quadratic constraints, resulting in  a generic convex program. The major complexity in solving such a program is due to the logarithm function in \eqref{eq:op2b}. Therefore, the use of modern convex solvers (e.g., SeDuMi \cite{Sturm} and MOSEK \cite{MOSEK}) becomes less efficient than standard ones. To bypass this issue, we use a lower bound of $\ln\bigl(1+\varphi_{n_{l}}\bigr)$ as \cite[Eq. (66)]{Dinh:JSAC:Dec2017}
\begin{align}\label{eq:op2blinear}
\ln\bigl(1+\varphi_{n_{l}}\bigr) \geq  \ln(1+\varphi^{(\kappa)}_{n_{l}})+\frac{\varphi^{(\kappa)}_{n_{l}}}{\varphi^{(\kappa)}_{n_{l}}+1} -\frac{(\varphi^{(\kappa)}_{n_{l}})^2}{\varphi^{(\kappa)}_{n_{l}}+1}\frac{1}{\varphi_{n_{l}}},\ \forall \varphi^{(\kappa)}_{n_{l}} > 0,  \varphi_{n_{l}} > 0,
\end{align}	
which is a concave function. We note that \eqref{eq:op2blinear} holds with equality at the optimum, i.e.,  $\varphi^{(\kappa)}_{n_{l}}=\varphi^{(\kappa+1)}_{n_{l}}$.
 Next, by introducing new variables $\bar{\boldsymbol{\varphi}}\triangleq \{\bar{\varphi}_{n_{l}}\}_{\forall n_l}$, the conic quadratic approximate program of \eqref{eq:convexprogramlog} is given as
 \begin{subequations}\label{eq:convexprogramCQP}
 	\begin{align}
 	& \underset{\hat{\boldsymbol{\rho}}, \mathbf{r}, \boldsymbol{\varphi},\bar{\boldsymbol{\varphi}},\boldsymbol{\varpi},\boldsymbol{\tau},\boldsymbol{\theta} }{\max} 
 	\quad \Bigl(1-\frac{\tau_p}{\tau_c}\Bigr) \sum\limits_{l \in \mathcal{L}} \sum\limits_{n_{l} \in \mathcal{N}_{l}} r_{n_{l}} \label{eq:eq:convexprogramCQP:a}\\ 
 	&\quad\ \textsf{s.t.} \qquad  \eqref{eq:op2c1b}{-}\eqref{eq:op2c1d},
 	\eqref{eq:op2c1eConvex}, 
 	\eqref{eq:op2d1b}{-}\eqref{eq:op2d1e},\eqref{eq:eq:convexprogramlog:c}, \eqref{eq:eq:convexprogramlog:d}, \label{eq:eq:convexprogramCQP:b}\\
 	&\qquad\qquad\  \mathcal{F}^{(\kappa)}(\varphi^{(\kappa)}_{n_{l}},\bar{\varphi}_{n_{l}}) \geq r_{n_{l}}\ln2,\ \forall n_{l} \in \mathcal{N}_l,\label{eq:eq:convexprogramCQP:c}\\
 	&\qquad\qquad\ 0.25\left(\varphi_{n_{l}}+\bar{\varphi}_{n_{l}}\right)^2 \geq 0.25\left(\varphi_{n_{l}}-\bar{\varphi}_{n_{l}}\right)^2  + 1,\ \forall n_{l} \in \mathcal{N}_l,\label{eq:eq:convexprogramCQP:d}
 	\end{align}
 \end{subequations}
 where $\mathcal{F}^{(\kappa)}(\varphi^{(\kappa)}_{n_{l}},\bar{\varphi}_{n_{l}})\triangleq \ln(1+\varphi^{(\kappa)}_{n_{l}})+\frac{\varphi^{(\kappa)}_{n_{l}}}{\varphi^{(\kappa)}_{n_{l}}+1} -\frac{(\varphi^{(\kappa)}_{n_{l}})^2}{\varphi^{(\kappa)}_{n_{l}}+1}\bar{\varphi}_{n_{l}}$. We note that  \eqref{eq:eq:convexprogramCQP:d} is a 
 second-order cone constraint and must hold with equality at the optimum. The proposed IA-based iterative algorithm is summarized in Algorithm \ref{alg_IterativeAlgorithm}.

\begin{algorithm}
	\begin{algorithmic}[1]
		\protect\caption{Proposed IA-based Iterative Algorithm to Solve Problem \eqref{eq:opp}.}
		\label{alg_IterativeAlgorithm}
		\global\long\def\algorithmicrequire{\textbf{Initialization:}}
		\REQUIRE  Set $\kappa:=0$ and   generate an initial feasible point $(\boldsymbol{\varpi}^{(0)},{\boldsymbol{\theta}}^{(0)},\boldsymbol{\varphi}^{(0)})$.
		\REPEAT
		\STATE Solve the conic quadratic approximate program \eqref{eq:convexprogramCQP} to obtain the optimal solution, denoted by ($\hat{\boldsymbol{\rho}}^{\star}, \mathbf{r}^{\star}, \boldsymbol{\varphi}^{\star},\bar{\boldsymbol{\varphi}}^{\star},\boldsymbol{\varpi}^{\star},\boldsymbol{\tau}^{\star},\boldsymbol{\theta}^{\star} $);
		\STATE Update ($ \boldsymbol{\varphi}^{(\kappa+1)},\boldsymbol{\varpi}^{(\kappa+1)},\boldsymbol{\theta}^{(\kappa+1)} ):=( \boldsymbol{\varphi}^{\star},\boldsymbol{\varpi}^{\star},\boldsymbol{\theta}^{\star} $);
		\STATE Set $\kappa:=\kappa+1;$
		\UNTIL Convergence, i.e., $\Bigl( \sum\limits_{l \in \mathcal{L}} \sum\limits_{n_{l} \in \mathcal{N}_{l}} r_{n_{l}}^{(\kappa)}- \sum\limits_{l \in \mathcal{L}} \sum\limits_{n_{l} \in \mathcal{N}_{l}} r_{n_{l}}^{(\kappa-1)}\Bigr)\Bigl/ \sum\limits_{l \in \mathcal{L}} \sum\limits_{n_{l} \in \mathcal{N}_{l}} r_{n_{l}}^{(\kappa-1)} < \epsilon$\\
		\STATE \textbf{Ouput}: $\boldsymbol{\rho}^{\star}$ with $\rho^{m,(\star)}_{n_{l}}=(\hat{\rho}^{m,(\star)}_{n_{l}})^2, \forall n_l \in \mathcal{N}_l$.
	\end{algorithmic} 
\end{algorithm}

\textit{Convergence and Complexity Analysis:} The proposed algorithm  starts by randomly generating an initial feasible point for the updated variables $(\boldsymbol{\varpi}^{(0)},{\boldsymbol{\theta}}^{(0)},\boldsymbol{\varphi}^{(0)})$. In each iteration, we solve the convex program \eqref{eq:convexprogramCQP} to produce the next feasible point $(\boldsymbol{\varphi}^{(\kappa+1)},\boldsymbol{\varpi}^{(\kappa+1)},\boldsymbol{\theta}^{(\kappa+1)} )$. This procedure is successively repeated until convergence, which is stated in the following
proposition.

\begin{proposition}\label{pro:1} Initialized from a feasible point $(\boldsymbol{\varpi}^{(0)},{\boldsymbol{\theta}}^{(0)},\boldsymbol{\varphi}^{(0)})$,
Algorithm \ref{alg_IterativeAlgorithm} produces a sequence \{$ \boldsymbol{\varphi}^{(\kappa)},\boldsymbol{\varpi}^{(\kappa)},\boldsymbol{\theta}^{(\kappa)} \}$ of improved solutions to problem \eqref{eq:convexprogramCQP}, which satisfy the Karush-Kuhn-Tucker (KKT)
conditions. In light of the IA principles, the sequence  $\Bigl\{\bigl(1-\frac{\tau_p}{\tau_c}\bigr) \sum\limits_{l \in \mathcal{L}} \sum\limits_{n_{l} \in \mathcal{N}_{l}} r_{n_{l}}^{(\kappa)}\Bigr\}_{\kappa=1}^{\infty}$ is monotonically
increasing and converges after a
finite number of iterations for a given error tolerance $\epsilon >0$.
\end{proposition}

\begin{proof}
	Please see Appendix \ref{app:B}.
\end{proof}

The computational complexity of Algorithm \ref{alg_IterativeAlgorithm} mainly depends on solving the approximate problem \eqref{eq:convexprogramCQP}, which is  polynomial in the number of
constraints and optimization variables. Problem \eqref{eq:convexprogramCQP} has $v=NM+3N+3\sum_{l=1}^L\frac{N_l(N_l-1)}{2}$ scalar real  variables and $c=8\sum_{l=1}^L\bigr(\frac{N_l(N_l-1)}{2}+M(N_l-1)\bigl)+M$ quadratic and linear
constraints. As a result, the worst-case computational cost of Algorithm \ref{alg_IterativeAlgorithm} in each iteration is $\mathcal{O}(v^2c^{2.5}+c^{3.5})$.

\section{Collocated massive MIMO-NOMA system} \label{sec:co}
In this section, we consider a COmMIMO-NOMA system, which serves as a benchmark for CFmMIMO-NOMA. The main differences between CFmMIMO-NOMA and COmMIMO-NOMA systems are as follows: $i$) in CFmMIMO-NOMA, in general $\beta_{m,n_{l}}$ $\neq$ $\beta_{{m'},n_{l}}$, for $m \neq {m'}$, whereas in COmMIMO-NOMA, $\beta_{m,n_{l}}$ $=$ $\beta_{{m'},n_{l}}$; and $ii$) in CFmMIMO-NOMA, a power constraint is applied at each AP individually, whereas in COmMIMO-NOMA, a total power constraint is applied at the collocated AP  equipped with $MK$ antennas. Unless otherwise specified, all notations and symbols given in the previous sections will be reused in this section.  

\subsection{Performance Analysis}
Similar to Lemma \ref{lemma1}, the closed-form expression for the SE of UE ${n_{l}}$ in the $l$-{th} cluster is given by
\begin{align}\label{eq:Rate1}
R^{\col}_{n_{l}} &= \Bigl(1-\frac{\tau_p}{\tau_c}\Bigr) \text{log}_{2}\bigl(1+\text{SINR}^{\col}_{n_{l}}\bigl) \nonumber \\
& = \Bigl(1-\frac{\tau_p}{\tau_c}\Bigl) \text{log}_{2}\bigl(1+\min_{{n'}_{l} = 1,\ldots,n_{l}}\text{SINR}^{n_{l},\col}_{{n'}_{l}}\bigl),\  \forall n_l \in \mathcal{N}_l. 
\end{align}

By replacing $\rho_{n_{l}}^m$ with $\rho_{n_{l}}, \forall n_l$,  $\text{SINR}^{n_{l},\col}_{n_{l}}$ and $\text{SINR}^{n_{l},\col}_{{n'}_{l}}$, $\forall {n'}_{l} < n_{l}$, are derived as follows:
\begin{align}\label{eq:SINR3}
\text{SINR}^{n_{l},\col}_{n_{l}} = \frac{(K-\tau_p) \rho_{n_{l}} \gamma_{n_{l}}}{\mathcal{I}^{n_{l}}_{{n}_{l}}(\boldsymbol{\rho}) + 1},\quad \text{and}\quad
\text{SINR}^{n_{l},\col}_{{n'}_{l}} = \frac{(K-\tau_p) \rho_{n_{l}} \gamma_{{n'}_{l}}}{\mathcal{I}^{n_{l}}_{{n'}_{l}}(\boldsymbol{\rho}) + 1},
\end{align}
where 
\begin{align}
\mathcal{I}^{n_{l}}_{{n}_{l}}(\boldsymbol{\rho})&\triangleq \sum\limits_{{n''}_{l}\in\mathcal{N}_{l}\setminus \{n_{l}\}} \eta_{{n''}_{l'}} (K-\tau_p) \rho_{{n''}_{l}} \gamma_{n_{l}} + \sum\limits_{{l'} \in \mathcal{L}} \sum\limits_{{n''}_{l'} \in \mathcal{N}_{l'}} \eta_{{n''}_{l'}} \rho_{{n''}_{l'}} \left(\beta_{n_{l}}-\gamma_{n_{l}}\right),\\
\mathcal{I}^{n_{l}}_{{n'}_{l}}(\boldsymbol{\rho})&\triangleq \sum\limits_{{n''}_{l}\in\mathcal{N}_{l}\setminus \{n_{l}\}} \eta_{{n''}_{l'}} (K-\tau_p) \rho_{{n''}_{l}} \gamma_{{n'}_{l}} + \sum\limits_{{l'} \in \mathcal{L}} \sum\limits_{{n''}_{l'} \in \mathcal{N}_{l'}} \eta_{{n''}_{l'}} \rho_{{n''}_{l'}} \left(\beta_{{n'}_{l}}-\gamma_{{n'}_{l}}\right),
\end{align}
and
 $\ds\gamma_{n_{l}} = \frac{\tau_p \rho_{n_l} \beta^{2}_{n_{l}}}{\tau_p \sum\limits_{{n'}_{l} \in \mathcal{N}_{l}} \rho_{{n'}_l} \beta_{{n'}_{l}} + 1}$;  $\eta_{{n''}_{l'}}$ is defined as 
\begin{equation}
\eta_{{n''}_{l'}} =
\begin{cases}
1, & \text{if} \ {l'} \ \neq \ l \ \text{or} \ {l'} \ = \ l \ \text{and} \ {n''}_{l} \ \leq \ {n}_{l},\\
\zeta_{{n}_{l}}, & \text{otherwise}.\\
\end{cases}       
\end{equation}

The SSE of  COmMIMO-NOMA system is expressed as follows:
\begin{align}
R^{\col}_{\Sigma} = \sum\limits_{l \in \mathcal{L}} \sum\limits_{n_{l} \in \mathcal{N}_{l}} R^{\col}_{n_{l}}= \Bigl(1-\frac{\tau_p}{\tau_c}\Bigr) \text{log}_{2}\bigl(1+\text{SINR}^{\col}_{n_{l}}\bigl).
\end{align}

The SSE maximization problem for COmMIMO-NOMA is stated as
\begin{subequations}\label{eqCol:opp}
	\begin{align}
	\underset{\boldsymbol{\rho}}{\max} \label{eqCol:op1a}
	& \quad\Bigl(1-\frac{\tau_p}{\tau_c}\Bigr) \sum\limits_{l \in \mathcal{L}} \sum\limits_{n_{l} \in \mathcal{N}_{l}}  \log_{2}\bigl(1+\mathrm{SINR}^{\col}_{n_{l}}\bigl) \\ \label{eqCol:op1b}
	\mathsf{s.t.}&\quad    \sum\limits_{l \in \mathcal{L}} \sum\limits_{n_{l} \in \mathcal{N}_l} \rho_{n_{l}} \leq P_{\max},\\ \label{eqCol:op1c}
	&\quad   \rho_{n_l} \leq \rho_{n_l+1}, n_l \in \left[1,N_l-1\right], \forall  l \in \mathcal{L}.
	\end{align}
\end{subequations}

\subsection{Proposed Solution to Problem \eqref{eqCol:opp} }
 
By making the change of variable as $\rho_{n_{l}}=(\hat{\rho}_{n_{l}})^2, \forall n_l \in \mathcal{N}_l$ and following similar steps from  \eqref{eq:op2} to \eqref{eq:convexprogramlog}, problem \eqref{eqCol:opp} is equivalently transformed to the following tractable form
\begin{subequations}\label{eqCol:convexprogramlog}
	\begin{align}
	& \underset{\hat{\boldsymbol{\rho}}, \mathbf{r}, \boldsymbol{\varphi},\boldsymbol{\theta} }{\max} 
	\quad \Bigl(1-\frac{\tau_p}{\tau_c}\Bigr) \sum\limits_{l \in \mathcal{L}} \sum\limits_{n_{l} \in \mathcal{N}_{l}} r_{n_{l}} \label{eqCol:eq:convexprogramlog:a}\\ 
	&\quad\ \textsf{s.t.} \quad\  \ln\left(1+\varphi_{n_{l}}\right) \geq r_{n_{l}}\ln2,\ \forall n_{l} \in \mathcal{N}_l,  \label{eqCol:eq:convexprogramlog:b}\\
		&\qquad\qquad\ \mathcal{I}^{n_{l}}_{{n'}_{l}}(\hat{\boldsymbol{\rho}}) \leq \theta^{n_{l}}_{{n'}_{l}},\ \forall {n'}_{l} < n_{l},\ \forall n_{l} \in \mathcal{N}_l  ,\label{eqCol:eq:convexprogramlog:b1}\\
		&\qquad\qquad\ \mathcal{I}^{n_{l}}_{{n}_{l}}(\hat{\boldsymbol{\rho}}) \leq \theta^{n_{l}}_{{n}_{l}},\  \forall n_{l} \in \mathcal{N}_l  ,\label{eqCol:eq:convexprogramlog:b2}\\	
	&\qquad\qquad\   \frac{(K-\tau_p) (\hat{\rho}_{n_{l}})^2 \gamma_{n'_{l}}}{\theta^{n_{l}}_{{n'}_{l}}+ 1} \geq \varphi_{n_{l}},\ \forall {n'}_{l} < n_{l},\ \forall n_{l} \in \mathcal{N}_l, \label{eqCol:eq:convexprogramlog:b3}\\	
		&\qquad\qquad\   \frac{(K-\tau_p) (\hat{\rho}_{n_{l}})^2 \gamma_{n_{l}}}{\theta^{n_{l}}_{{n}_{l}} + 1} \geq \varphi_{n_{l}},\  \forall n_{l} \in \mathcal{N}_l	,\label{eqCol:eq:convexprogramlog:b4}\\	
	&\qquad\qquad\    \sum\limits_{l \in \mathcal{L}} \sum\limits_{n_{l} \in \mathcal{N}_l} (\hat{\rho}_{n_{l}})^2 \leq P_{\max},\label{eqCol:eq:convexprogramlog:c}\\ 
	&\qquad\qquad\   \hat{\rho}_{n_l} \leq \hat{\rho}_{n_l+1}, n_l \in \left[1,N_l-1\right], \forall   l \in \mathcal{L}, \label{eqCol:eq:convexprogramlog:d}  
	\end{align}
\end{subequations}
where
\begin{align}
\mathcal{I}^{n_{l}}_{{n'}_{l}}(\hat{\boldsymbol{\rho}})&\triangleq \sum\limits_{{n''}_{l}\in\mathcal{N}_{l}\setminus \{n_{l}\}} \eta_{{n''}_{l'}} (K-\tau_p) (\hat{\rho}_{{n''}_{l}})^2 \gamma_{{n'}_{l}} + \sum\limits_{{l'} \in \mathcal{L}} \sum\limits_{{n''}_{l'} \in \mathcal{N}_{l'}} \eta_{{n''}_{l'}} (\hat{\rho}_{{n''}_{l'}})^2 \left(\beta_{{n'}_{l}}-\gamma_{{n'}_{l}}\right),\nonumber\\
\mathcal{I}^{n_{l}}_{{n}_{l}}(\hat{\boldsymbol{\rho}})&\triangleq  \sum\limits_{{n''}_{l}\in\mathcal{N}_{l}\setminus \{n_{l}\}} \eta_{{n''}_{l'}} (K-\tau_p) (\hat{\rho}_{{n''}_{l}})^2 \gamma_{n_{l}} + \sum\limits_{{l'} \in \mathcal{L}} \sum\limits_{{n''}_{l'} \in \mathcal{N}_{l'}} \eta_{{n''}_{l'}} (\hat{\rho}_{{n''}_{l'}})^2 \left(\beta_{n_{l}}-\gamma_{n_{l}}\right).   \nonumber
\end{align}

The nonconvex constraints are \eqref{eqCol:eq:convexprogramlog:b3} and \eqref{eqCol:eq:convexprogramlog:b4}. Let $(\hat{\rho}_{n_{l}}^{(\kappa)},\theta^{n_{l},(\kappa)}_{{n}_{l}})$ be a feasible point of $(\hat{\rho}_{n_{l}},\theta^{n_{l}}_{{n}_{l}})$ at iteration $\kappa$. By \eqref{eq:op2blinear}, 
 the conic quadratic approximate program for solving \eqref{eqCol:convexprogramlog} is given as
\begin{subequations}\label{eqCol:convexprogramCQP}
	\begin{align}
	& \underset{\hat{\boldsymbol{\rho}}, \mathbf{r}, \boldsymbol{\varphi},\bar{\boldsymbol{\varphi}},\boldsymbol{\theta} }{\max} 
	\quad \Bigl(1-\frac{\tau_p}{\tau_c}\Bigr) \sum\limits_{l \in \mathcal{L}} \sum\limits_{n_{l} \in \mathcal{N}_{l}} r_{n_{l}} \label{eqCol:eq:convexprogramCQP:a}\\ 
	&\quad\ \textsf{s.t.} \quad\  \eqref{eq:eq:convexprogramCQP:c}, \eqref{eq:eq:convexprogramCQP:d}, \eqref{eqCol:eq:convexprogramlog:b1}, \eqref{eqCol:eq:convexprogramlog:b2}, \eqref{eqCol:eq:convexprogramlog:c}, \eqref{eqCol:eq:convexprogramlog:d}, \\
	&\qquad\qquad\  (K-\tau_p)\gamma_{n'_{l}}\mathcal{G}^{(\kappa)}(\hat{\rho}_{n_{l}},\theta^{n_{l}}_{{n'}_{l}}) \geq \varphi_{n_{l}},\ \forall {n'}_{l} < n_{l},\ \forall n_{l} \in \mathcal{N}_l,\label{eqCol:eq:convexprogramCQP:c}\\
    &\qquad\qquad\  (K-\tau_p)\gamma_{n_{l}}\mathcal{G}^{(\kappa)}(\hat{\rho}_{n_{l}},\theta^{n_{l}}_{{n}_{l}}) \geq \varphi_{n_{l}},\  \forall n_{l} \in \mathcal{N}_l,\label{eqCol:eq:convexprogramCQP:d}
	\end{align}
\end{subequations}
where $\ds \mathcal{G}^{(\kappa)}(\hat{\rho}_{n_{l}},\theta^{n_{l}}_{{n'}_{l}})\triangleq \frac{2\hat{\rho}_{n_{l}}^{(\kappa)} }{{\theta^{n_{l},(\kappa)}_{{n'}_{l}}}+1}\hat{\rho}_{n_{l}} -\frac{\bigl(\hat{\rho}_{n_{l}}^{(\kappa)}\bigr)^2}{\bigl({\theta^{n_{l},(\kappa)}_{{n'}_{l}}}+1\bigr)^2}(\theta^{n_{l}}_{{n'}_{l}}+1)$ and $\ds \mathcal{G}^{(\kappa)}(\hat{\rho}_{n_{l}},\theta^{n_{l}}_{{n}_{l}})\triangleq \frac{2\hat{\rho}_{n_{l}}^{(\kappa)} }{{\theta^{n_{l},(\kappa)}_{{n}_{l}}}+1}\hat{\rho}_{n_{l}} -\frac{\bigl(\hat{\rho}_{n_{l}}^{(\kappa)}\bigr)^2}{\bigl({\theta^{n_{l},(\kappa)}_{{n}_{l}}}+1\bigr)^2}(\theta^{n_{l}}_{{n}_{l}}+1)$. The solution to problem \eqref{eqCol:opp} can be found by using Algorithm \ref{alg_IterativeAlgorithm}, in which we replace problem \eqref{eq:convexprogramCQP} by problem \eqref{eqCol:convexprogramCQP} in Step 2. The worst-case computational complexity of solving \eqref{eqCol:convexprogramCQP} in each iteration is $\mathcal{O}(\bar{v}^2\bar{c}^{2.5}+\bar{c}^{3.5})$, where $\bar{v}=4N+\sum_{l=1}^L\frac{N_l(N_l-1)}{2}$  and $\bar{c}=\sum_{l=1}^L\bigl(N_l(N_l-1)+\frac{(N_l-1)^2}{2}\bigl)+2N+1$ are scalar real  variables and 
constraints, respectively.


\section{Numerical Results}\label{sec}
We now quantitatively assess the performance of the proposed unsupervised ML-based UC algorithms in CFmMIMO-NOMA system. 
\begin{table}[t]
	\begin{minipage}{0.5\linewidth}
		\centering
		\fontsize{11}{12}\selectfont
		\captionof{table}{Simulation Parameters.}
		\label{tab: parameter}
		\vspace{-5pt}
		\scalebox{1}{
			\begin{tabular}{l|l}
				\hline
				Parameter & Value \\
				\hline\hline
				Reference distances ($d_0$, $d_{1}$) & (10,50) m \\
				System bandwidth ($B$) & 20 MHz \\
				Number of APs ($M$) & 32 \\ 
				Number of UEs ($N$) & 10 \\
				Number of antennas per AP ($K$) & 8\\ 
				Total power budget for all APs & 40 dBm \\
				Power budget at UEs & 23 dBm \\
				Noise power at receivers & -104 dBm \\
				SIC performance coefficient at UEs & 0.05 \\
				\hline		   				
			\end{tabular}
		}
	\end{minipage}
	\hfill
	\begin{minipage}{0.44\linewidth}
		\centering
		\includegraphics[width=1.2\columnwidth,trim={0cm 0cm 0cm 0cm}]{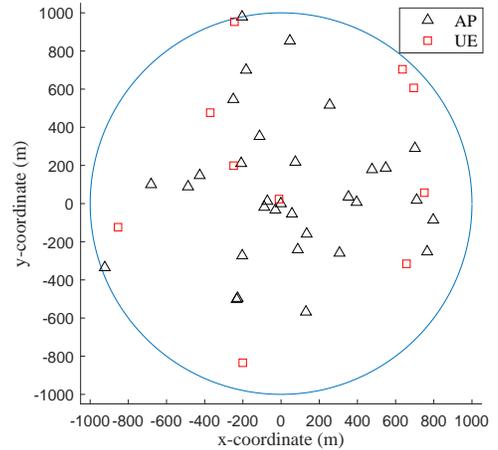}
		\vspace{-10pt}
		\captionof{figure}{A system topology with  $M=32$ APs and $N=10$ UEs is used in numerical examples.}
		\label{fig: Layout}
	\end{minipage}
\end{table}
\subsection{Simulation Parameters}
A CFmMIMO-NOMA system including $M=32$ APs and $N=10$ UEs is considered as shown in Fig. \ref{fig: Layout}, where all APs and UEs are uniformly distributed within a circular region with a radius of 1 km. The large-scale fading of all channels is modeled as \cite{Ngo:TWC:Mar2017} $\ds\beta_{m,n_{l}} =\ 10^{\frac{{\rm{PL}}(d_{m,n_{l}})+\sigma_{sh}z}{10}}$,\ $\forall m \in \mathcal{M}$, $n_{l} \in \mathcal{N}_{l}$, where $d_{m,n_{l}}$ is the distance from  AP$_{m}$ to UE ${n_{l}}$. The shadow fading is modeled as an RV $z$, which follows $\mathcal{CN} (0,1)$  with standard deviation $\sigma_{sh}$ = 8 dB. The three-slope path loss model is considered  as \cite{Ngo:TWC:Mar2017,TangVTC01,HieuJSAC2020}
\begin{align}\label{eq:PL}
{\rm{PL}}(d_{m,n_{l}}) &= -140.7 -35 {\rm{log}}_{10} (d_{m,n_{l}}) + 20 a_0 {\rm{log}}_{10} \Bigl(\frac{d_{m,n_{l}}}{d_0}\Bigr) + 15 a_1 {\rm{log}}_{10} \Bigl(\frac{d_{m,n_{l}}}{d_1}\Bigr),
\end{align}
where $d_j$, with $j = \left\{0,1\right\}$, represents the reference distance and $a_j = {{\rm{max}}\left\{0,\frac{d_i - d_{m,n_{l}}}{|d_i - d_{m,n_{l}}|}\right\}}$. Note that ${\rm{PL}}(d_{m,n_{l}})$ in \eqref{eq:PL} is measured in dB, while all distances are in km. Unless otherwise stated, other key parameters are shown in Table \ref{tab: parameter}, where all APs are assumed to have the same power budget \cite{HieuJSAC2020,Ngo:TWC:Mar2017}. The used convex solver is SeDuMi \cite{Sturm} in the MATLAB environment.

\subsection{Selection of the Number of Clusters L}
\begin{table}[hbt!]
		\centering
		\captionof{table}{Silhouette Score for CFmMIMO-NOMA and COmMIMO-NOMA.}
		\label{table1}
		\vspace{-5pt}
		\scalebox{0.82}{
	\begin{tabular}{c|c|c|c|c|c|c|c|c|c}
	\hline
	\multicolumn{2}{l|}{ Number of clusters $L$
	} & 2 & 3 & 4 & 5 & 6 & 7 & 8 & 9  \\ \hline
	\multirow{2}{*}{Silhouette Score}  & CFmMIMO-NOMA & 0.86 & 0.08 & 0.52 & \textbf{0.97} & 0.23 & 0.36 & 0.52 & 0.89  \\ \cline{2-10} 
	& COmMIMO-NOMA  & 0.90 & 0.85 & 0.67 & \textbf{0.99} & 0.65 & 0.78 & 0.88 & 0.93 \\ \hline
	\end{tabular}}
\end{table}

The performance of the k-means based UC algorithms is highly affected by the value of number of clusters $L$ \cite{CuiTWC2018,PalouGC2018}. Thus, it is essential to investigate the particular feature of the UEs' distribution in CFmMIMO-NOMA system to choose a proper number of clusters, such that the SSE is maximized. A reliable and precise approach to validate the optimal number of clusters $L$ is the silhouette score \cite{Geron2019}, which is the mean silhouette coefficient of all UEs. The silhouette coefficient of an UE is calculated as $\ds\frac{c-b}{{\rm{max}}(c,b)}$, where $b$ denotes the mean distance to other UEs in the same cluster (so-called the mean intra-cluster distance), and $c$ represents the mean distance to UEs of the next closest cluster which is the one that minimizes $b$, excluding the UE’s own cluster (so-called mean nearest-cluster distance). The value of the silhouette coefficient ranges from -1 to +1. A coefficient close to +1 means that the UE is well matched to its own cluster and far from other clusters. A coefficient close to 0 indicates that the UE is near a cluster boundary, whereas a coefficient close to -1 implies that the UE is assigned to the wrong cluster. Table \ref{table1} shows the silhouette score versus the number of clusters $L$. It is observed that the optimal number of clusters for this setting is $L^\star=5$. 

In what follows, we set $L=5$ to verify the performance analysis in Section \ref{sec:FPA} and to evaluate the performance of the proposed algorithms in Section \ref{sec:OPA}. 

\subsection{Numerical Results for the Performance Analysis} \label{sec:FPA}

\begin{figure}[t]
	\begin{minipage}{0.46\linewidth}
		\centering
		\includegraphics[width=\columnwidth,trim={0cm 0cm 0cm 0cm}]{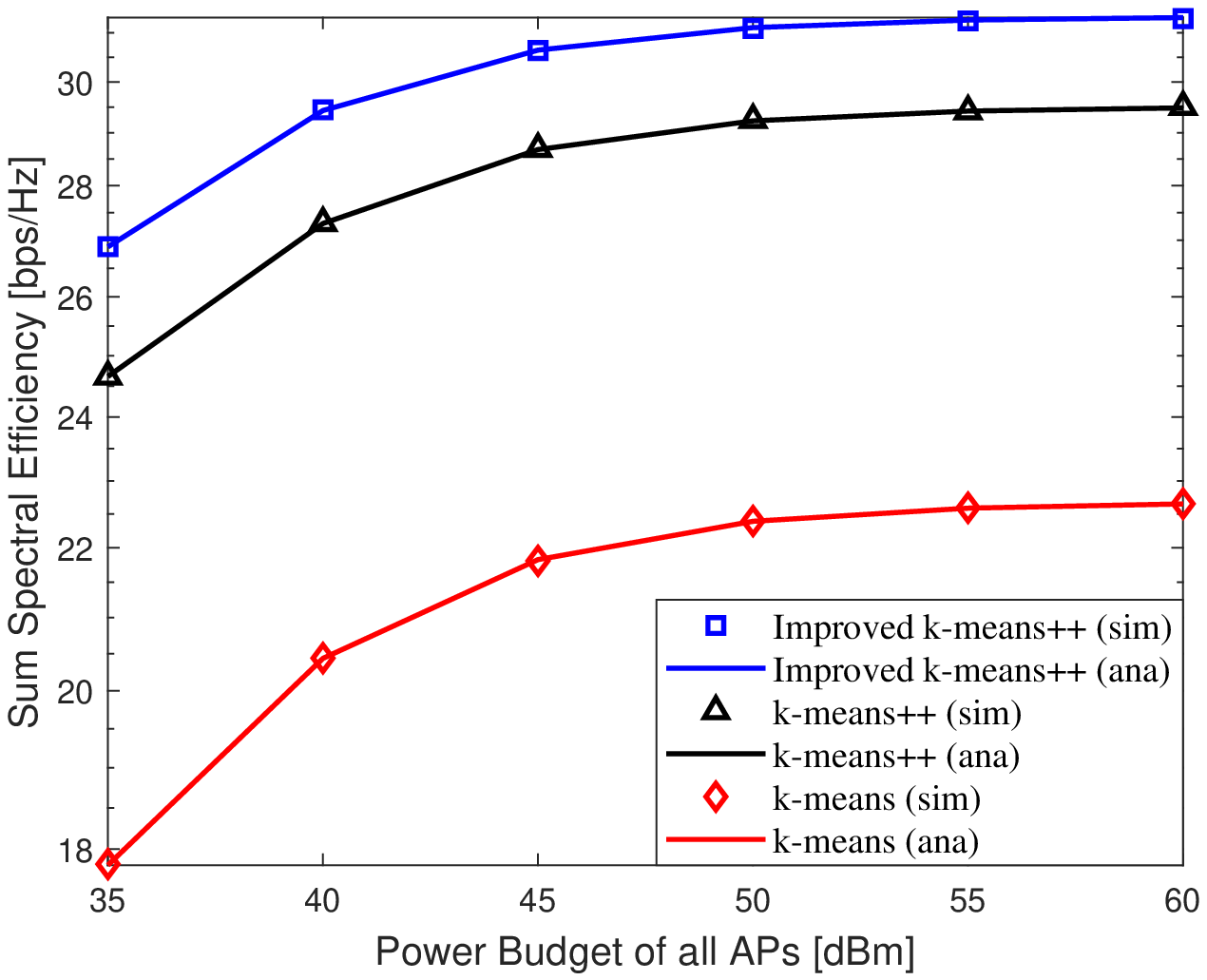}
		\vspace{-5pt}
		\caption{The SSE of CFmMIMO-NOMA  versus the total power budget of all APs for the k-means, k-means++, and  improved k-means++ algorithms.}
		\label{fig:FSSE2}
	\end{minipage}
	\hfill
	\begin{minipage}{0.46\linewidth}
		\centering
		\includegraphics[width=1.02\columnwidth,trim={0cm 0cm 0cm 0cm}]{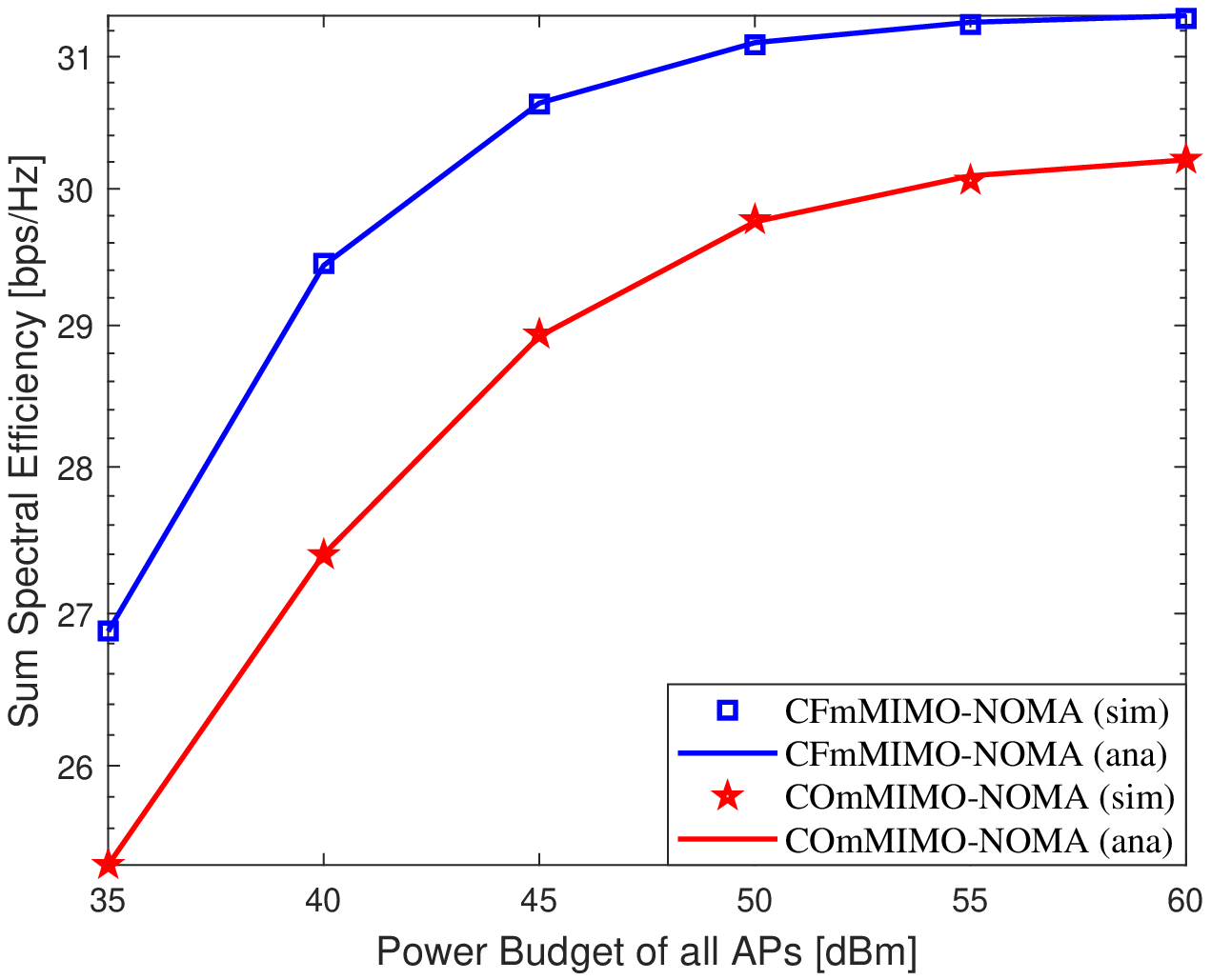}
		\vspace{-5pt}
		\caption{The SSE of CFmMIMO-NOMA and COmMIMO-NOMA versus the total power budget of all APs.}
		\label{fig:FSSE}
	\end{minipage}
\end{figure}

We now investigate the performance of the two proposed unsupervised ML-based UC algorithms with fixed PA. The transmit power at each AP allocated to a specific UE follows the fixed PA scheme. Each AP allocates equal power to each cluster, and then, the fractional transmit PA \cite{BenjebbourGC2014} is used to allocate the power to a specific UE in each cluster based on the virtual channel gains presented in Section \ref{sec:PA}. As a benchmark, we also consider the COmMIMO-NOMA system, which is presented in Section \ref{sec:co}. 

 Fig. \ref{fig:FSSE2} illustrates the SSE performance of CFmMIMO-NOMA versus the total power budget of all APs for the proposed UC algorithms. For comparison, the performance of the k-means algorithm is also plotted. It can be seen that the proposed UC algorithms significantly outperform the conventional k-means one. On the other hand, the improved k-means++ achieves the best SSE among all algorithms. This further confirms the importance of the effective  initialization of centroids that improves the quality of the grouping process; otherwise, the use of NOMA becomes less efficient.
 Next, the SSE performance of the CFmMIMO-NOMA and COmMIMO-NOMA systems using the improved k-means++ algorithm versus the total power budget of all APs is shown in Fig. \ref{fig:FSSE}. We can observe that the performance of the CFmMIMO-NOMA system is  better than that of COmMIMO-NOMA. This is attributed to the fact that CFmMIMO with many distributed APs brings the service antennas closer to UEs which not only reduces path losses but also
 provides  higher degree of macro-diversity, compared to COmMIMO. Further, from Figs. \ref{fig:FSSE2} and \ref{fig:FSSE}, simulation results are well matched with the derived closed-form expressions of SSE in Section \ref{sec:PA},  verifying the correctness of our analytical results. In the following numerical results, unless otherwise specified, the improved k-means++ algorithm is used for UC.

\subsection{Numerical Results for Optimal Power Allocation (Algorithm \ref{alg_IterativeAlgorithm})} \label{sec:OPA}
\begin{figure}[t]
	\centering
	\vspace{-10pt}
	\includegraphics[width=0.5\columnwidth,trim={0cm 0.0cm 0.0cm 0cm}]{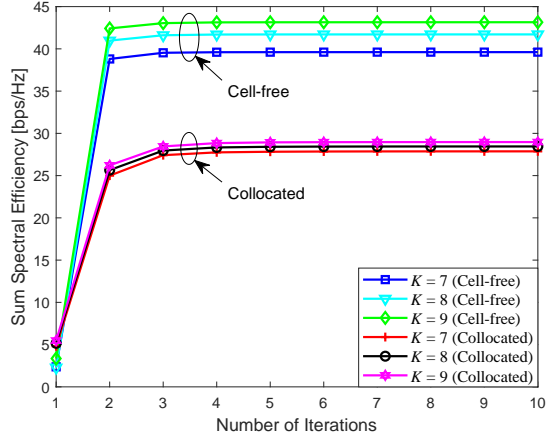}
	\vspace{-5pt}
	\caption{Convergence behavior of Algorithm \ref{alg_IterativeAlgorithm}  with different number of AP antennas, $K$.}
	\label{fig:OSSE}
\end{figure}  

In Fig. \ref{fig:OSSE}, we  evaluate the convergence speed of Algorithm \ref{alg_IterativeAlgorithm} for CFmMIMO-NOMA and COmMIMO-NOMA with different values of $K$.
The proposed algorithm converges within three iterations and the convergence speed of both systems is  not sensitive to
the number of AP antennas, $K$. As expected, the SSE is monotonically
increasing after each iteration. Compared to the results in Figs. \ref{fig:FSSE2} and \ref{fig:FSSE} with fixed PA at the power budget of 40 dBm,  Algorithm \ref{alg_IterativeAlgorithm} yields a significantly better performance in terms of SSE. The results demonstrate the effectiveness of the proposed algorithm to achieve the optimal SSE.

\begin{figure}[t]
	\begin{minipage}{0.46\linewidth}
		\centering
		\includegraphics[width=\columnwidth,trim={0cm 0cm 0cm 0cm}]{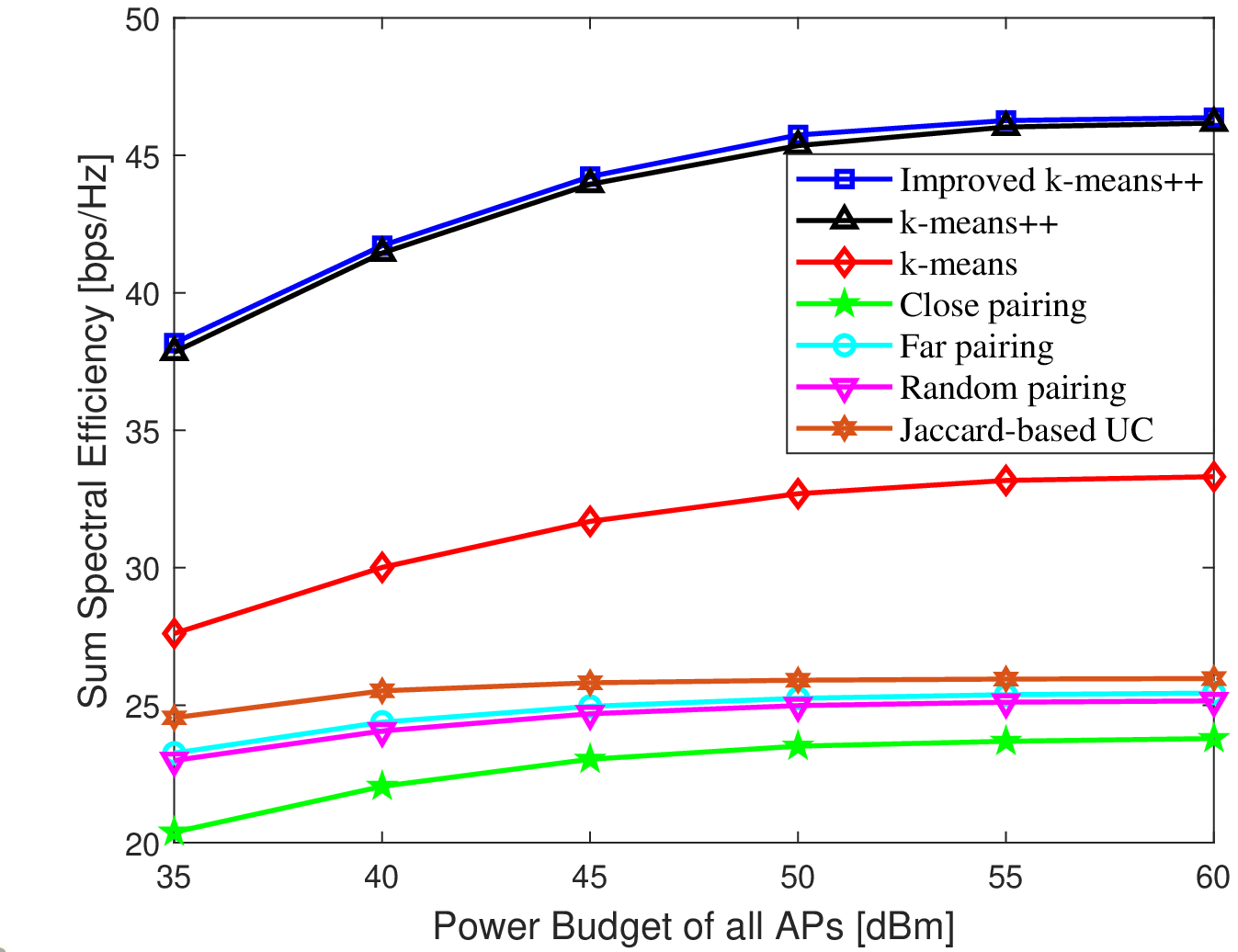}
		\vspace{-5pt}
		\caption{The SSE of different UC algorithms.}
		\label{fig:FSSE3}
	\end{minipage}
	\hfill
	\begin{minipage}{0.46\linewidth}
		\centering
		\includegraphics[width=1.01\columnwidth,trim={0cm 0cm 0cm 0cm}]{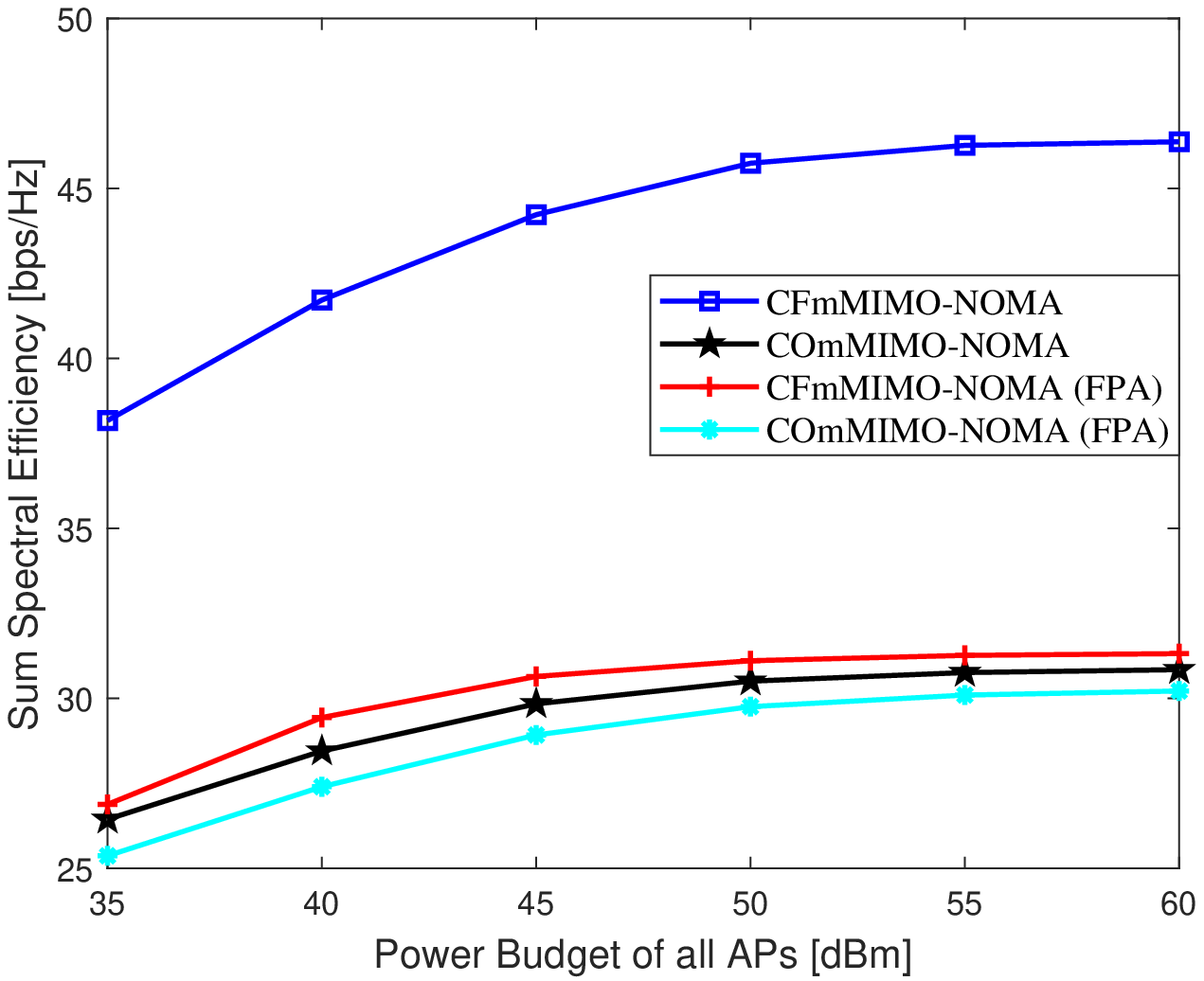}
		\vspace{-5pt}
		\caption{SSE of CFmMIMO-NOMA and COmMIMO-NOMA: with and without PA.}
		\label{fig:OSSE1}
	\end{minipage}
\end{figure}

Fig. \ref{fig:FSSE3} shows the impact of the proposed  k-means++ and improved k-means++ algorithms on the system performance of CFmMIMO-NOMA. For comparison, we also plot the SSE of the  k-means (i.e., Algorithm 1) and the recently proposed UC approaches, including near pairing, far pairing, random pairing \cite{BasharTCOM2020}, and the Jaccard-based UC \cite{RezaeiCL2020}. The main result observed from the figure is that   the proposed unsupervised ML-based UC algorithms 
achieve better SSE performance compared to the baseline ones, and the performance gaps are wider when $P_{\max}$ increases. This implies that the two proposed UC schemes are capable of exploiting UC more effectively, so that the SSE is remarkably enhanced. In Fig. \ref{fig:OSSE1}, we demonstrate the benefit of optimizing PA for CFmMIMO-NOMA and COmMIMO-NOMA systems. The SSE of both systems is significantly enhanced with optimal PA compared to the fixed PA scheme. Hence, this shows the necessity of optimizing PA for both systems, especially for the CFmMIMO-NOMA system.

\begin{figure}[t]
	\begin{minipage}{0.46\linewidth}
		\centering
		\includegraphics[width=1.0\columnwidth,trim={0cm 0cm 0cm 0cm}]{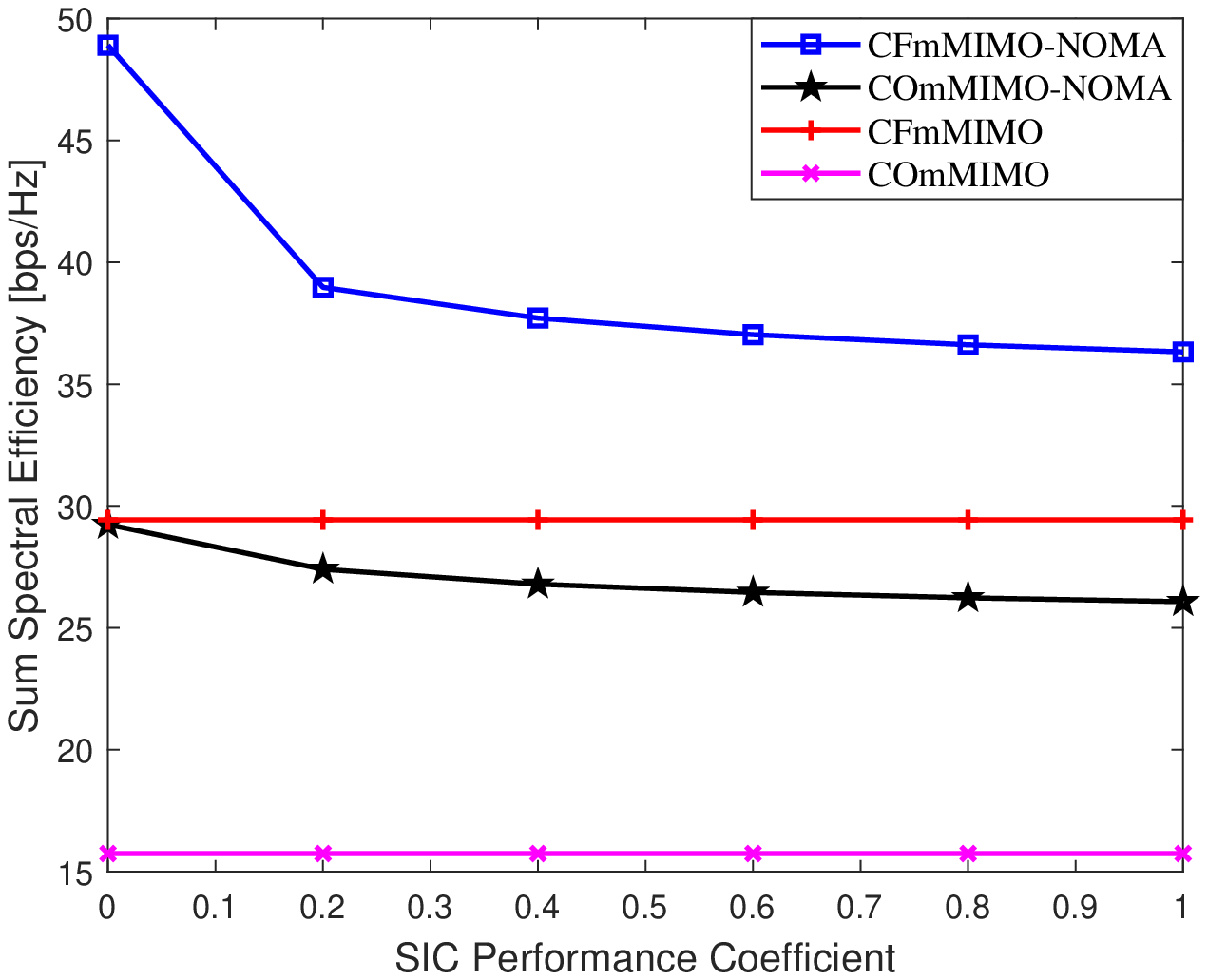}
		\vspace{-5pt}
		\caption{The effect of SIC performance coefficient on the SSE of CFmMIMO-NOMA and COmMIMO-NOMA systems.}
		\label{fig:OSSE2}
	\end{minipage}
	\hfill
	\begin{minipage}{0.46\linewidth}
		\centering
		\includegraphics[width=1.02\columnwidth,trim={0cm 0cm 0cm 0cm}]{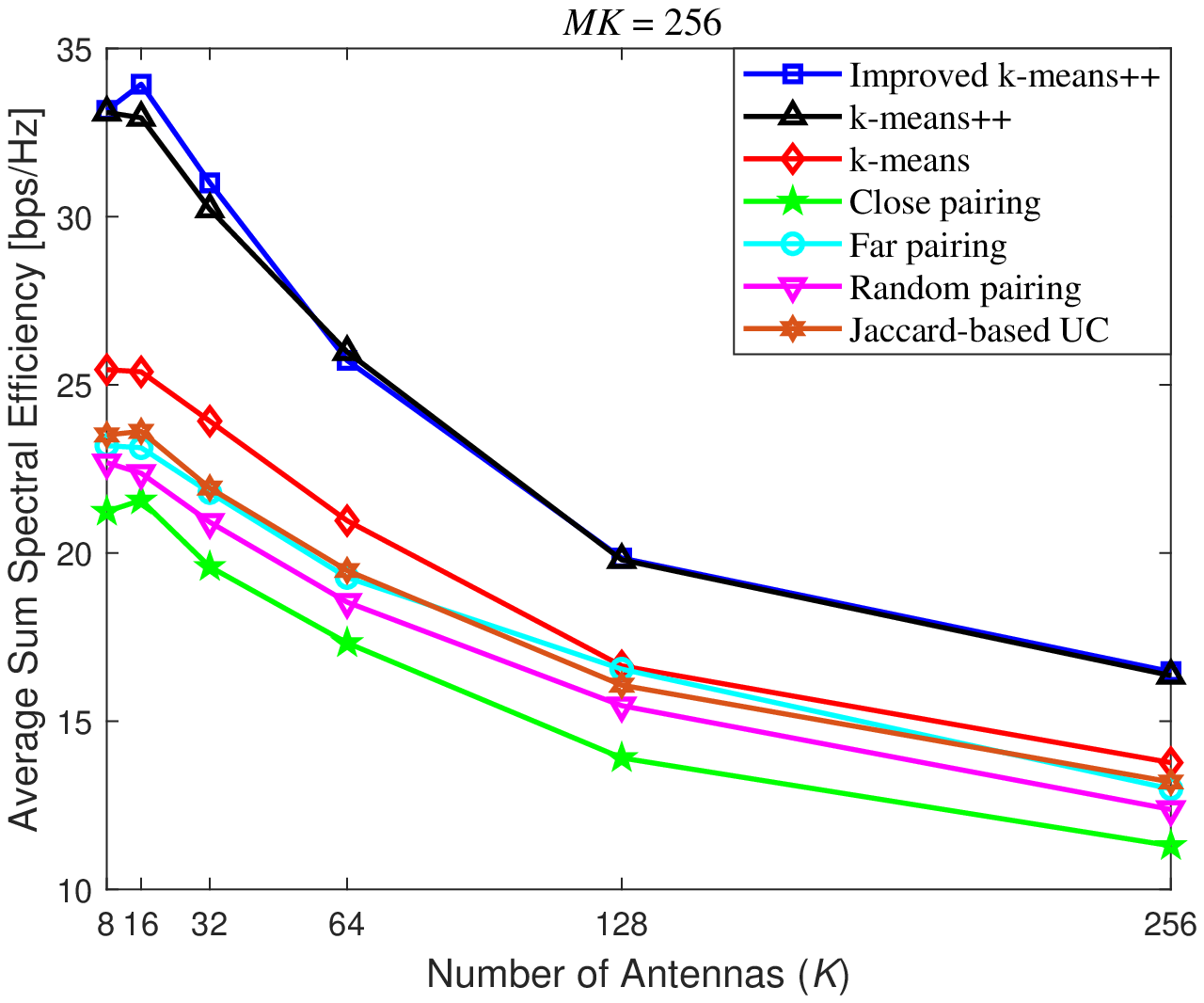}
		\vspace{-5pt}
		\caption{The joint effect of the numbers of antennas $K$ and  APs $M$ on the average SSE of different UC algorithms.}
		\label{fig:OSSE3}
	\end{minipage}
\end{figure}

Next, the effect of the SIC performance coefficient $\zeta_{{n}_{l}}$ on the SSE of CFmMIMO-NOMA and COmMIMO-NOMA is examined in Fig. \ref{fig:OSSE2}. We note that $\zeta_{{n}_{l}} = 1$ ($\zeta_{{n}_{l}} = 0$) indicates no SIC (perfect SIC), while $0 < \zeta_{{n}_{l}} < 1$ means imperfect SIC. The system performance without NOMA/SIC is plotted. It is clear that the SSE of CFmMIMO-NOMA degrades when $\zeta_{{n}_{l}},\forall n_l$ increases. It implies that the SIC performance coefficient requires to be small enough to exploit the full potential of NOMA in CFmMIMO.  Nevertheless, the SSE achieved by CFmMIMO-NOMA and COmMIMO-NOMA systems is much higher than their counterparts without NOMA/SIC.

Finally, we investigate the joint effect of the numbers of antennas $K$ and  APs $M$ on the average SSE of different UC algorithms. We fix $MK=256$ and select $K$ from the set $K\in[8,16,32,64,128,256]$. When $K=256$, then $M=1$, which represents COmMIMO-NOMA. From the figure, we see that the SSE first increases and then decreases when $K$ increases. This result reveals an interesting insight: for extremely small $K$, the use of fpZF is less efficient in terms of canceling  inter-cluster interference. However, the higher the value of $K$, the lower the value of APs $M$. This not only increases path losses, but also reduces the degree of macro-diversity. The results suggest that the optimal value of $(M,K)$ can improve the SSE of  CFmMIMO-NOMA, e.g., $(M,K)=(16,16)$ for improved k-means++ and $(M,K)=(32,8)$ for  k-means++ in this setting.

\section{Conclusion}\label{sec:con}
In this paper, we have investigated downlink CFmMIMO-NOMA system, where  two
efficient unsupervised ML-based UC algorithms are developed to effectively cluster users into disjoint clusters. Using the fpZF precoding  at APs, we have derived  closed-form expressions for the SSE of CFmMIMO-NOMA, taking into account  effects of intra-cluster pilot contamination, inter-cluster interference, and imperfect SIC. In addition, we
have considered the problem of power allocation to maximize SSE. Since  the formulated problem is intractable, we have developed a low-complexity iterative algorithm based on the IA  framework for its solution. Numerical results have confirmed the effectiveness of the proposed UC algorithms, and show their superior
performance  compared to the baseline schemes. The proposed PA algorithm converges fast, and  significantly outperforms CFmMIMO-NOMA without optimizing PA and
COmMIMO-NOMA in terms of SSE.

\appendices
\section{Proof of Lemma 1}\label{ap1}
\textit{1) Computation of} $|\text{DS}_{n_{l}}|^2$\textit{:} By using \eqref{eq:e} and \eqref{eq:hmnl3}, the numerator in \eqref{eq:SINR} is rewritten as
\begin{align}\label{eq:DS1}
|\text{DS}_{n_{l}}|^2 = \Bigl|\mathbb{E} \Bigl\{\sum\limits_{m\in\mathcal{M}} \sqrt{\rho^{m}_{n_{l}}} \textbf{h}_{m,n_{l}}^H \boldsymbol{w}_{m,l}\Bigl\}\Bigl|^2 
= \Bigl|\mathbb{E} \Bigl\{\sum\limits_{m\in\mathcal{M}} \sqrt{\rho^{m}_{n_{l}}} \hat{\textbf{h}}_{m,n_{l}}^H \boldsymbol{w}_{m,l}\Bigl\}\Bigl|^2 
= (K-\tau_p) \Bigl(\sum\limits_{m\in\mathcal{M}} \sqrt{\rho^{m}_{n_{l}} \gamma_{m,n_{l}}}\Bigl)^2,
\end{align}
where the second equality is obtained due to the independence between the estimation error vector $\textbf{e}_{m,n_{l}}$ and the channel estimate $\hat{\textbf{h}}_{m,n_{l}}$.

\textit{2) Computation of} $\mathbb{E} \left\{|\text{BU}_{n_{l}}|^2\right\}$\textit{:} The first term of the denominator in \eqref{eq:SINR} is reformulated as
\begin{align}\label{eq:BU1}
\mathbb{E} \left\{|\text{BU}_{n_{l}}|^2\right\} &= \mathbb{E} \Bigl\{\Bigl|\Bigl(\sum\limits_{m\in\mathcal{M}} \sqrt{\rho^{m}_{n_{l}}} \textbf{h}_{m,n_{l}}^H \boldsymbol{w}_{m,l} - \mathbb{E} \Bigl\{\sum\limits_{m\in\mathcal{M}} \sqrt{\rho^{m}_{n_{l}}} \textbf{h}_{m,n_{l}}^H \boldsymbol{w}_{m,l}\Bigl\}\Bigl)\Bigl|^2\Bigl\} \nonumber \\
& = \mathbb{E} \Bigl\{\Bigl|\sum\limits_{m\in\mathcal{M}} \sqrt{\rho^{m}_{n_{l}}} \textbf{h}_{m,n_{l}}^H \boldsymbol{w}_{m,l}\Bigl|^2\Bigl\} - \Bigl|\mathbb{E} \Bigl\{\sum\limits_{m\in\mathcal{M}} \sqrt{\rho^{m}_{n_{l}}} \textbf{h}_{m,n_{l}}^H \boldsymbol{w}_{m,l}\Bigl\}\Bigl|^2.
\end{align}

According to \eqref{eq:e} and \eqref{eq:hmnl3}, the first term in \eqref{eq:BU1} is further derived as follows:
\begin{align}\label{eq:BU2}
\mathbb{E} \Bigl\{\Bigl|\sum\limits_{m\in\mathcal{M}} \sqrt{\rho^{m}_{n_{l}}} \textbf{h}_{m,n_{l}}^H \boldsymbol{w}_{m,l}\Bigl|^2\Bigl\} &= \mathbb{E} \Bigl\{\Bigl|\sum\limits_{m\in\mathcal{M}} \sqrt{\rho^{m}_{n_{l}}} \hat{\textbf{h}}_{m,n_{l}}^H \boldsymbol{w}_{m,l}\Bigl|^2\Bigl\} + \mathbb{E} \Bigl\{\Bigl|\sum\limits_{m\in\mathcal{M}} \sqrt{\rho^{m}_{n_{l}}} \textbf{e}_{m,n_{l}}^H \boldsymbol{w}_{m,l}\Bigl|^2\Bigl\} \nonumber \\
&= (K-\tau_p) \Bigl(\sum\limits_{m\in\mathcal{M}} \sqrt{\rho^{m}_{n_{l}} \gamma_{m,n_{l}}}\Bigl)^2 + \sum\limits_{m\in\mathcal{M}} \rho^{m}_{n_{l}} \Bigl(\beta_{m,n_{l}}-\gamma_{m,n_{l}}\Bigl).
\end{align}

Substituting \eqref{eq:DS1} and \eqref{eq:BU2} into \eqref{eq:BU1}, \eqref{eq:BU1} can be rewritten as
\begin{align}\label{eq:BU3}
\mathbb{E} \Bigl\{|\text{BU}_{n_{l}}|^2\Bigl\} &= (K-\tau_p) \Bigl(\sum\limits_{m\in\mathcal{M}} \sqrt{\rho^{m}_{n_{l}} \gamma_{m,n_{l}}}\Bigl)^2 + \sum\limits_{m\in\mathcal{M}} \rho^{m}_{n_{l}} \left(\beta_{m,n_{l}}-\gamma_{m,n_{l}}\right) \nonumber \\
& - (K-\tau_p) \Bigl(\sum\limits_{m\in\mathcal{M}} \sqrt{\rho^{m}_{n_{l}} \gamma_{m,n_{l}}}\Bigl)^2  \nonumber \\
& = \sum\limits_{m\in\mathcal{M}} \rho^{m}_{n_{l}} \left(\beta_{m,n_{l}}-\gamma_{m,n_{l}}\right).
\end{align}

\textit{3) Computation of} $\sum\limits_{{n'}_{l}=1}^{n_{l} - 1} \mathbb{E} \left\{|\text{ICI}_{n_{l}}|^2\right\}$\textit{:} Based on \eqref{eq:BU2}, the second term of the denominator in \eqref{eq:SINR} is computed as
\begin{align}\label{eq:ICI1}
\sum\limits_{{n'}_{l}=1}^{n_{l} - 1} \mathbb{E} \left\{|\text{ICI}_{n_{l}}|^2\right\} &= \sum\limits_{{n'}_{l}=1}^{n_{l} - 1} \mathbb{E} \Bigl\{\Bigl|\sum\limits_{m\in\mathcal{M}} \sqrt{\rho^{m}_{{n'}_{l}}} \textbf{h}_{m,n_{l}}^H \boldsymbol{w}_{m,l}\Bigl|^2\Bigl\} \nonumber \\
& = \sum\limits_{{n'}_{l}=1}^{n_{l} - 1} (K-\tau_p) \Bigl(\sum\limits_{m\in\mathcal{M}} \sqrt{\rho^{m}_{{n'}_{l}} \gamma_{m,n_{l}}}\Bigl)^2 + \sum\limits_{{n'}_{l}=1}^{n_{l} - 1} \sum\limits_{m\in\mathcal{M}} \rho^{m}_{{n'}_{l}} \left(\beta_{m,n_{l}}-\gamma_{m,n_{l}}\right).
\end{align}

\textit{4) Computation of} $\sum\limits_{{n''}_{l} = n_{l} + 1}^{N_l} \mathbb{E} \left\{|\text{RICI}_{n_{l}}|^2\right\}$\textit{:} According to \eqref{eq:BU2}, the third term of the denominator in \eqref{eq:SINR} is rewritten as
\begin{align} \label{eq:RICI1}
\sum\limits_{{n''}_{l} = n_{l} + 1}^{N_l} \mathbb{E} \left\{|\text{RICI}_{n_{l}}|^2\right\} &= \sum\limits_{{n''}_{l} = n_{l} + 1}^{N_l} \mathbb{E} \Bigl\{\Bigl|\sqrt{\zeta_{{n}_{l}}} \sum\limits_{m\in\mathcal{M}} \sqrt{\rho^{m}_{{n''}_{l}}} \textbf{h}_{m,n_{l}}^H \boldsymbol{w}_{m,l} x_{{n''}_{l}} \Bigl|^2\Bigl\} \nonumber \\
& = \sum\limits_{{n''}_{l} = n_{l} + 1}^{N_l} \zeta_{{n}_{l}} (K-\tau_p) \Bigl(\sum\limits_{m\in\mathcal{M}} \sqrt{\rho^{m}_{{n''}_{l}} \gamma_{m,n_{l}}}\Bigl)^2 \nonumber \\
&+ \sum\limits_{{n''}_{l} = n_{l} + 1}^{N_l} \sum\limits_{m\in\mathcal{M}} \zeta_{{n}_{l}} \rho^{m}_{{n''}_{l}} \left(\beta_{m,n_{l}}-\gamma_{m,n_{l}}\right).
\end{align}

\textit{5) Computation of} $\mathbb{E} \left\{|\text{UI}_{n_{l}}|^2\right\}$\textit{:} From \eqref{eq:BU2}, the fourth term of the denominator in \eqref{eq:SINR} is shown as follows:
\begin{align} \label{eq:UI2}
\mathbb{E} \left\{|\text{UI}_{n_{l}}|^2\right\} &= \sum\limits_{{l'}\in\mathcal{L}\setminus \{l\}} \sum\limits_{n_{l'}=1}^{N_{l'}} \mathbb{E} \Bigl\{\Bigl|\sum\limits_{m\in\mathcal{M}} \sqrt{\rho^{m}_{{n}_{l'}}} \textbf{h}_{m,n_{l}}^H \boldsymbol{w}_{m,{l'}} \Bigl|^2\Bigl\} \nonumber \\
& = \sum\limits_{{l'}\in\mathcal{L}\setminus \{l\}} \sum\limits_{n_{l'}=1}^{N_{l'}} \mathbb{E} \Bigl\{\Bigl|\sum\limits_{m\in\mathcal{M}} \sqrt{\rho^{m}_{{n}_{l'}}} \textbf{e}_{m,n_{l}}^H \boldsymbol{w}_{m,{l'}} \Bigl|^2\Bigl\} \nonumber \\
& = \sum\limits_{{l'}\in\mathcal{L}\setminus \{l\}} \sum\limits_{n_{l'}=1}^{N_{l'}} \sum\limits_{m\in\mathcal{M}} \rho^{m}_{{n}_{l'}} \left(\beta_{m,n_{l}}-\gamma_{m,n_{l}}\right),
\end{align}
where the second equality in \eqref{eq:UI2} is obtained due to the property of the fpZF precoding.

Finally, by substituting \eqref{eq:DS1}, \eqref{eq:BU3}, \eqref{eq:ICI1},  \eqref{eq:RICI1}, and \eqref{eq:UI2} into \eqref{eq:SINR}, $\text{SINR}^{n_{l}}_{n_{l}}$ is obtained as in \eqref{eq:SINR1}. Following the similar steps for deriving $\text{SINR}^{n_{l}}_{n_{l}}$, $\text{SINR}^{n_{l}}_{{n'}_{l}}$ can be easily derived as in \eqref{eq:SINR2}.

\section{Proof of Proposition \ref{pro:1}}\label{app:B}
By contradiction and IA principles, we can easily prove that constraints \eqref{eq:op2c1b}{-}\eqref{eq:op2c1d},
\eqref{eq:op2c1eConvex}, 
\eqref{eq:op2d1b}{-}\eqref{eq:op2d1e} and \eqref{eq:eq:convexprogramCQP:d} must hold with equality at optimum. Let us define $\mathcal{F}(\varphi_{n_{l}})\triangleq \ln(1+\varphi_{n_{l}})$. From \eqref{eq:op2blinear}, we have
\begin{equation}
\mathcal{F}(\varphi_{n_{l}}) \geq \mathcal{F}^{(\kappa)}(\varphi^{(\kappa)},\bar{\varphi}_{n_{l}}),
\end{equation}
and 
\begin{equation}
\mathcal{F}(\varphi_{n_{l}}^{(\kappa)}) = \mathcal{F}^{(\kappa)}(\varphi^{(\kappa)},\bar{\varphi}_{n_{l}}).
\end{equation}

Thus, it is true that
\begin{equation}
\mathcal{F}(\varphi_{n_{l}}^{(\kappa)}) \geq \mathcal{F}^{(\kappa-1)}(\varphi^{(\kappa)},\bar{\varphi}_{n_{l}}) \geq \mathcal{F}^{(\kappa-1)}(\varphi^{(\kappa-1)},\bar{\varphi}_{n_{l}}) =\mathcal{F}(\varphi_{n_{l}}^{(\kappa-1)}). 
\end{equation}

These results imply that $(\boldsymbol{\varpi}^{(\kappa)},{\boldsymbol{\theta}}^{(\kappa)},\boldsymbol{\varphi}^{(\kappa)})$ is an improved solution to problem \eqref{eq:convexprogramCQP}, compared to $(\boldsymbol{\varpi}^{(\kappa-1)},$ ${\boldsymbol{\theta}}^{(\kappa-1)},\boldsymbol{\varphi}^{(\kappa-1)})$. By \cite[Theorem 1]{Marks:78}, the sequence $\{\boldsymbol{\varpi}^{(\kappa)},{\boldsymbol{\theta}}^{(\kappa)},\boldsymbol{\varphi}^{(\kappa)}\}$ converges to at least local optima which satisfy the KKT conditions. As a result, the objective value of problem \eqref{eq:convexprogramCQP} is monotonically
increasing, i.e., $\bigl(1-\frac{\tau_p}{\tau_c}\bigr) \sum\limits_{l \in \mathcal{L}} \sum\limits_{n_{l} \in \mathcal{N}_{l}} r_{n_{l}}^{(\kappa)} \geq \bigl(1-\frac{\tau_p}{\tau_c}\bigr) \sum\limits_{l \in \mathcal{L}} \sum\limits_{n_{l} \in \mathcal{N}_{l}} r_{n_{l}}^{(\kappa-1)}$. In addition, the sequence of the objective values is upper bounded  due
to power constraints \eqref{eq:eq:convexprogramlog:c}, which completes the proof.

\begingroup
\setstretch{1.3}
\bibliographystyle{IEEEtran}
\bibliography{IEEEfull}

\begin{thebibliography}{10}
\providecommand{\url}[1]{#1}
\csname url@samestyle\endcsname
\providecommand{\newblock}{\relax}
\providecommand{\bibinfo}[2]{#2}
\providecommand{\BIBentrySTDinterwordspacing}{\spaceskip=0pt\relax}
\providecommand{\BIBentryALTinterwordstretchfactor}{4}
\providecommand{\BIBentryALTinterwordspacing}{\spaceskip=\fontdimen2\font plus
\BIBentryALTinterwordstretchfactor\fontdimen3\font minus
  \fontdimen4\font\relax}
\providecommand{\BIBforeignlanguage}[2]{{%
\expandafter\ifx\csname l@#1\endcsname\relax
\typeout{** WARNING: IEEEtran.bst: No hyphenation pattern has been}%
\typeout{** loaded for the language `#1'. Using the pattern for}%
\typeout{** the default language instead.}%
\else
\language=\csname l@#1\endcsname
\fi
#2}}
\providecommand{\BIBdecl}{\relax}
\BIBdecl

\bibitem{Iotanalytics2016}
\BIBentryALTinterwordspacing
Iot-analytics.com, \emph{State of the IoT 2018: Number of IoT devices now at 7B
  – Market accelerating}, Aug. 2018. [Online]. Available:
  \url{https://iot-analytics.com/state-of-the-iot-update-q1-q2-2018-number-of-iot-devices-now-7b}
\BIBentrySTDinterwordspacing

\bibitem{Cisco2017_CVNI}
\BIBentryALTinterwordspacing
\emph{Cisco Visual Networking Index: Global Mobile Data Traffic Forecast
  Update, 2016-2021}, Mar. 2017. [Online]. Available:
  \url{https://www.cisco.com/c/en/us/solutions/collateral/service-provider/visual-networking-index-vni/mobile-white-paper-c11-520862.html}
\BIBentrySTDinterwordspacing

\bibitem{Islam:COMSurTutor:2017}
S.~M.~R. Islam, N.~Avazov, O.~A. Dobre, and K.-S. Kwak, ``Power-domain
  non-orthogonal multiple access {(NOMA)} in {5G} systems: Potentials and
  challenges,'' \emph{IEEE Commun. Surveys Tuts.}, vol.~19, no.~2, pp.
  721--742, 2nd Quart. 2017.

\bibitem{Dinh:JSAC:Dec2017}
V.-D. Nguyen, H.~D. Tuan, T.~Q. Duong, H.~V. Poor, and O.-S. Shin, ``Precoder
  design for signal superposition in {MIMO-NOMA} multicell networks,''
  \emph{IEEE J. Select. Areas Commun.}, vol.~35, no.~12, pp. 2681--2695, Dec.
  2017.

\bibitem{HieuTCOM2019}
H.~V. {Nguyen}, V.-D. {Nguyen}, O.~A. {Dobre}, D.~N. {Nguyen}, E.~{Dutkiewicz},
  and O.-S. {Shin}, ``Joint power control and user association for {NOMA}-based
  full-duplex systems,'' \emph{IEEE Trans. Commun.}, vol.~67, no.~11, pp.
  8037--8055, Nov. 2019.

\bibitem{Ngo:TWC:Mar2017}
H.~Q. {Ngo}, A.~{Ashikhmin}, H.~{Yang}, E.~G. {Larsson}, and T.~L. {Marzetta},
  ``Cell-free massive {MIMO} versus small cells,'' \emph{IEEE Trans. Wireless
  Commun.}, vol.~16, no.~3, pp. 1834--1850, Mar. 2017.

\bibitem{Nayebi:IEEETWC:Jul2017}
E.~{Nayebi}, A.~{Ashikhmin}, T.~L. {Marzetta}, H.~{Yang}, and B.~D. {Rao},
  ``Precoding and power optimization in cell-free massive {MIMO} systems,''
  \emph{IEEE Trans. Wireless Commun.}, vol.~16, no.~7, pp. 4445--4459, July
  2017.

\bibitem{Bashar:IEEETWC:Apr2019}
M.~{Bashar}, K.~{Cumanan}, A.~G. {Burr}, M.~{Debbah}, and H.~Q. {Ngo}, ``On the
  uplink max–min {SINR} of cell-free massive {MIMO} systems,'' \emph{IEEE
  Trans. Wireless Commun.}, vol.~18, no.~4, pp. 2021--2036, Apr. 2019.

\bibitem{Ngo:IEEETGCN:Mar2018}
H.~Q. {Ngo}, L.-N. {Tran}, T.~Q. {Duong}, M.~{Matthaiou}, and E.~G. {Larsson},
  ``On the total energy efficiency of cell-free massive {MIMO},'' \emph{IEEE
  Trans. Green Commun. Netw.}, vol.~2, no.~1, pp. 25--39, Mar. 2018.

\bibitem{ChenTWC2018}
Z.~{Chen} and E.~{Björnson}, ``Channel hardening and favorable propagation in
  cell-free massive {MIMO} with stochastic geometry,'' \emph{IEEE Trans.
  Commun.}, vol.~66, no.~11, pp. 5205--5219, Nov. 2018.

\bibitem{chen2020massive}
\BIBentryALTinterwordspacing
X.~Chen, D.~W.~K. Ng, W.~Yu, E.~G. Larsson, N.~Al-Dhahir, and R.~Schober,
  ``Massive access for {5G} and beyond,'' Feb. 2020. [Online]. Available:
  \url{https://arxiv.org/abs/2002.03491}
\BIBentrySTDinterwordspacing

\bibitem{LiWCL2018}
Y.~{Li} and G.~A.~A. {Baduge}, ``{NOMA}-aided cell-free massive {MIMO}
  systems,'' \emph{IEEE Wireless Commun. Lett.}, vol.~7, no.~6, pp. 950--953,
  Dec. 2018.

\bibitem{ZhangICC2018}
Y.~{Zhang}, H.~{Cao}, M.~{Zhou}, and L.~{Yang}, ``Spectral efficiency
  maximization for uplink cell-free massive {MIMO-NOMA} networks,'' in
  \emph{2019 IEEE Inter. Conf. Commun. Works. (ICC Workshops)}, May 2019, pp.
  1--6.

\bibitem{RezaeiTWC2020}
F.~{Rezaei}, C.~{Tellambura}, A.~A. {Tadaion}, and A.~R. {Heidarpour}, ``Rate
  analysis of cell-free massive {MIMO-NOMA} with three linear precoders,''
  \emph{IEEE Trans. Commun.}, vol.~68, no.~6, pp. 3480--3494, June 2020.

\bibitem{Islam:IEEEWirelessComm:Apr2018}
S.~M.~R. Islam, M.~Zeng, O.~A. Dobre, and K.-S. Kwak, ``Resource allocation for
  downlink {NOMA} systems: Key techniques and open issues,'' \emph{IEEE
  Wireless Commun.}, vol.~25, no.~2, pp. 40--47, Apr. 2018.

\bibitem{BasharTCOM2020}
M.~{Bashar}, K.~{Cumanan}, A.~G. {Burr}, H.~Q. {Ngo}, L.~{Hanzo}, and
  P.~{Xiao}, ``On the performance of cell-free massive {MIMO} relying on
  adaptive {NOMA/OMA} mode-switching,'' \emph{IEEE Trans. Commun.}, vol.~68,
  no.~2, pp. 792--810, Feb. 2020.

\bibitem{RezaeiCL2020}
F.~{Rezaei}, A.~R. {Heidarpour}, C.~{Tellambura}, and A.~A. {Tadaion},
  ``Underlaid spectrum sharing for cell-free massive {MIMO-NOMA},'' \emph{IEEE
  Commun. Lett.}, vol.~24, no.~4, pp. 907--911, Apr. 2020.

\bibitem{HeTWC2017}
R.~{He}, Q.~{Li}, B.~{Ai}, Y.~L.-A. {Geng}, A.~F. {Molisch}, V.~{Kristem},
  Z.~{Zhong}, and J.~{Yu}, ``A kernel-power-density-based algorithm for channel
  multipath components clustering,'' \emph{IEEE Trans. Wireless Commun.},
  vol.~16, no.~11, pp. 7138--7151, Nov. 2017.

\bibitem{XieICC2017}
X.~{Xie}, Z.~{Zhang}, H.~{Jiang}, J.~{Dang}, and L.~{Wu}, ``Cluster-based
  geometrical dynamic stochastic model for {MIMO} scattering channels,'' in
  \emph{Proc. Inter. Conf. Wireless Commun. and Signal Process. (WCSP)}, Oct.
  2017, pp. 1--5.

\bibitem{WangACCESS2019}
Y.~{Wang}, A.~{Liu}, X.~{Xia}, and K.~{Xu}, ``Exploiting the clustered sparsity
  for channel estimation in hybrid analog-digital massive {MIMO} systems,''
  \emph{IEEE Access}, vol.~7, pp. 4989--5000, Dec. 2018.

\bibitem{RozarioICoTCT2018}
A.~B. {Rozario} and M.~F. {Hossain}, ``An architecture for {M}2{M}
  communications over cellular networks using clustering and hybrid
  {TDMA-NOMA},'' in \emph{Proc. Inter. Conf. Infor. and Commun. Tech.
  (ICoICT)}, May 2018, pp. 18--23.

\bibitem{i14}
A.~K. Jain, ``Data clustering: 50 years beyond k-means,'' \emph{Pattern
  Recognition Lett.}, vol.~31, no.~8, pp. 651--666, June 2010.

\bibitem{CabreraITNAC2018}
E.~{Cabrera} and R.~{Vesilo}, ``An enhanced k-means clustering algorithm with
  non-orthogonal multiple access {(NOMA)} for {MMC} networks,'' in \emph{Proc.
  Inter. Telecommun Net. and App. Conf. (ITNAC)}, Nov. 2018, pp. 1--8.

\bibitem{CuiTWC2018}
J.~{Cui}, Z.~{Ding}, P.~{Fan}, and N.~{Al-Dhahir}, ``Unsupervised machine
  learning-based user clustering in millimeter-wave-{NOMA} systems,''
  \emph{IEEE Trans. Wireless Commun.}, vol.~17, no.~11, pp. 7425--7440, Sep.
  2018.

\bibitem{PalouGC2018}
F.~{Riera-Palou}, G.~{Femenias}, A.~G. {Armada}, and A.~{Pérez-Neira},
  ``Clustered cell-free massive {MIMO},'' in \emph{Proc. IEEE Globecom
  Workshops (GC Wkshps)}, Dec. 2018, pp. 1--6.

\bibitem{CespedesTCOM2020}
M.~{Morales-Céspedes}, O.~A. {Dobre}, and A.~{García-Armada}, ``Semi-blind
  interference aligned {NOMA} for downlink {MU}-{MISO} systems,'' \emph{IEEE
  Trans. Commun.}, vol.~68, no.~3, pp. 1852--1865, Mar. 2020.

\bibitem{Marks:78}
B.~R. Marks and G.~P. Wright, ``A general inner approximation algorithm for
  nonconvex mathematical programs,'' \emph{Oper. Res.}, vol.~26, no.~4, pp.
  681--683, July-Aug. 1978.

\bibitem{InterdonatoGCSIP2018}
G.~{Interdonato}, M.~{Karlsson}, E.~{Bjornson}, and E.~G. {Larsson}, ``Downlink
  spectral efficiency of cell-free massive {MIMO} with full-pilot
  zero-forcing,'' in \emph{IEEE Global Conf. Signal and Infor. Process.
  (GlobalSIP)}, Nov. 2018, pp. 1003--1007.

\bibitem{HieuJSAC2020}
H.~V. {Nguyen}, V.-D. {Nguyen}, O.~A. {Dobre}, S.~K. {Sharma},
  S.~{Chatzinotas}, B.~{Ottersten}, and O.-S. {Shin}, ``On the spectral and
  energy efficiencies of full-duplex cell-free massive {MIMO},'' \emph{IEEE J.
  Select. Areas Commun.}, pp. 1--1, June 2020.

\bibitem{ArthurDA2007}
D.~{Arthur} and S.~{Vassilvitskii}, ``K-means++: The advantages of careful
  seeding,'' in \emph{Proc. Symp. Discrete Algorithms}, Jan. 2007, pp.
  1027--1035.

\bibitem{FrantiPR2019}
P.~{Fränti} and S.~{Sieranoja}, ``How much can k-means be improved by using
  better initialization and repeats?'' \emph{Pattern Recognit.}, vol.~93, pp.
  95--112, Sep. 2019.

\bibitem{BachemAI2016}
O.~{Bachem}, M.~{Lucic}, S.~H. {Hassani}, and A.~{Krause}, ``Approximate
  k-means++ in sublinear time,'' in \emph{Proc. 30th AAAI Conf. Artif.
  Intell.}, Feb. 2016, pp. 1459--1467.

\bibitem{Tulino2004}
A.~M. {Tulino} and S.~{Verdú}, ``Random matrix theory and wireless
  communications,'' \emph{Commun. Inf. Theory}, vol.~1, no.~1, pp. 1--182, June
  2004.

\bibitem{Beck:JGO:10}
A.~Beck, A.~Ben-Tal, and L.~Tetruashvili, ``A sequential parametric convex
  approximation method with applications to nonconvex truss topology design
  problems,'' \emph{J. Global Optim.}, vol.~47, no.~1, pp. 29--51, May 2010.

\bibitem{Sturm}
J.~F. Sturm, ``Using sedumi 1.02, a {MATLAB} toolbox for optimization over
  symmetric cones,'' \emph{Optimiz. Methods and Softw.}, vol. 11-12, pp.
  625--653, Sep. 1999.

\bibitem{MOSEK}
\BIBentryALTinterwordspacing
``I.~{MOSEK} aps,'' 2014. [Online]. Available: \url{http://www.mosek.com}
\BIBentrySTDinterwordspacing

\bibitem{TangVTC01}
A.~{Tang}, J.~{Sun}, and K.~{Gong}, ``Mobile propagation loss with a low base
  station antenna for {NLOS} street microcells in urban area,'' in \emph{Proc.
  IEEE Veh. Tech. Conf. (VTC Spring)}, May 2001, pp. 333--336.

\bibitem{Geron2019}
A.~{Geron}, ``Hands-on machine learning with scikit-learn and tensorflow:
  Concepts, tools, and techniques to build intelligent systems,''
  \emph{O'Reilly Media Inc.}, pp. 1--484, Sep. 2019.

\bibitem{BenjebbourGC2014}
A.~{Benjebbour}, A.~{Li}, Y.~{Kishiyama}, H.~{Jiang}, and T.~{Nakamura},
  ``System-level performance of downlink {NOMA} combined with {SUMIMO} for
  future {LTE} enhancements,'' in \emph{Proc. IEEE Globecom Workshops (GC
  Wkshps)}, Dec. 2014, pp. 706--710.

\end{thebibliography}
\endgroup

\end{document}